\documentclass[11pt]{article}
\usepackage{amsmath,amssymb,epsf,amsfonts,euscript} 
\usepackage{amsthm}
\usepackage{enumerate}
\usepackage{latexsym}
\usepackage{graphicx}
\usepackage{graphics}
\usepackage{amsfonts}
\usepackage[usenames]{color}
\usepackage[font={footnotesize}]{caption}
\usepackage{xcolor}
\newcounter{alphthm}
\newtheorem{defn}{Definition}
\newtheorem{prop}{Proposition}
\newtheorem{cor}{Corollary}

\newtheorem{rem}{Remark}

\newcommand{\be}{\begin{equation}}
\newcommand{\ee}{\end{equation}}
\newcommand{\ben}{\begin{enumerate}}
\newcommand{\een}{\end{enumerate}}

\renewcommand{\theequation}
{\arabic{equation}}
\def\beq{\begin{equation}}
\def\eeq{\end{equation}}

\renewcommand{\theequation}{\arabic{section}.\arabic{equation}}
\title{\Large\hspace{-2cm}  Multi-scale Invariant Fields: Estimation and Prediction}

\author{\large \hspace{-1.5cm}$\mathrm{H.\: Ghasemi}^{a},\:$ $ \mathrm{ S.\: Rezakhah}^{a}$\footnote{Address correspondence to Saeid  Rezakhah,  Faculty of Mathematics and Computer Science, Amirkabir University of Technology, Tehran, Iran
E-mail: rezakhah@aut.ac.ir}
 \vspace{.5cm}$$
$ \:\mathrm{ N.\:  Modarresi}^{b}$ \\\\
\small{\em { $^{a}$ Faculty of Mathematics and Computer Science, Amirkabir University of Technology, Tehran, Iran.}}\\
 \hspace{-0.8cm}\small{ \em { $^{b}$ Faculty of Mathematical Sciences and Computer, Allameh Tabataba'i University, Tehran, Iran.}}}\vspace{-1mm}

\date{}
\begin{document}
\maketitle

\begin{abstract}
\noindent
Extending the concept of multi-selfsimilar random field we study multi-scale invariant (MSI) fields which have  component-wise discrete scale invariant property.
Assuming scale parameters as $\lambda_i>1$, $i=1,\ldots,d$ and the parameter space as $(1, \infty)^d$, the first scale rectangle is referred to
the rectangle $ (1, \lambda_1)\times \ldots \times (1, \lambda_d)$.
We show that the covariance function  of the sampled Markov MSI field are characterized by the variances and covariances of samples
inside first scale rectangle.
As an example of MSI field, a two-dimensional simple fractional Brownian  sheet (sfBs) is demonstrated.
Also real data of  the precipitation in some area of Brisbane in Australia
for two days (25 and  26  January 2013) are examined. We show that  precipitation on this area has MSI property and estimate it as a simple MSI field with stationary increments inside scale intervals.
{\color{black} This structure enables us to predict the precipitation in surface and time. We apply the mean absolute percentage error as a measure for the accuracy of  the predictions.}
\\ \\ \\
{\it Mathematics Subject Classification MSC 2010:} 60G18; 60G22; 62M15; 62H05.
\end{abstract}

\section{Introduction}
Gaussian self-similar fields have been extensively studied and applied in various area as hydrology, biology, economics, finance and image processing
\cite{Sam}.
    In probability theory, a random field is a family of random variables indexed in
a multi-dimensional space.\\
As stated by Genton et al. \cite{Taqqu}, a random field $\{X(\mathbf{t}), \mathbf{t}\in\mathbb{R}^d\}$ is said to be multi-selfsimilar (MSS) if
for some Hurst vector $\mathbf{H}=(H_1,\cdots , H_d)>{\bf 0}$ and any $\boldsymbol{\Lambda}=(\lambda_1, \cdots , \lambda_d)>{\bf 1}$
\begin{equation}\label{eq14}
{\bf X}({\bf\Lambda}\circ{\bf t})\overset{\mathcal{L}}{=}\big(\prod_{i=1}^d \lambda_i^{H_i}\big){\bf X(t)}.
\end{equation}
where $\mathbf{H}=(H_1, \ldots, H_d)$, $\boldsymbol{\Lambda}=(\lambda_1, \ldots, \lambda_d)$ and $\overset{\mathcal{L}}{=}$ denotes equality of
finite dimensional distributions,
 and $\circ$ is the Hadamard product that operates as $\boldsymbol{\Lambda}\circ \mathbf{t}=(\lambda_1t_1,\ldots,\lambda_dt_d)$.
The random field is said to be multi-scale invariant (MSI) of index $\mathbf{H}$ and scale  $\boldsymbol{\Lambda}^{\prime}
=(\lambda^{\prime}_1,\ldots,\lambda^{\prime}_d)$, if (\ref{eq14}) holds for some
case $\boldsymbol{\Lambda}=\boldsymbol{\Lambda}^{\prime}$.
In one-dimensional, the discrete scale invariant (DSI) process initially studied by Borgnat et al. \cite{Borgnat} is a scale invariance or
self-similar process only for specific choice of scale parameter. Balasis et al \cite {Bala1, Bala2, Bala3} and  Bartolozzi et al \cite{Bala4} have studied wide range of applications of DSI processes in Dynamic Magnetosphere, DST time series and stock markets.  \\

Let $\{X(\mathbf{t}), \mathbf{t}\in [1, \infty)^d \}$ be some MSI field with prescribed scale vector $\boldsymbol{\Lambda}
=(\lambda_1,\ldots,\lambda_d)$ where $\lambda_i$'s are greater than one.
Extending the method of Modarresi and Rezakhah \cite{R-M1,R-M2}, we consider component-wise geometric sampling of the field at points
$\boldsymbol{\alpha}^{\mathbf{k}}=(\alpha_1^{k_1},\ldots, \alpha_d^{k_d})^T$, $k_1, k_2, \ldots, k_d\in\mathbb{N}_0=\{0,1,\ldots\}$,
to get $\{X(\boldsymbol{\alpha}^{\mathbf{k}}), \boldsymbol{\alpha}^{\mathbf{k}}\in {\Bbb R}^d\}$ as the sampled MSI field,
where $\alpha_i$'s are determined by $\lambda_i=\alpha_i^{n_i}$ for  $ i=1, 2, \ldots, d$  while $n_1, \ldots, n_d$ are arbitrary positive integers.
So we have $\prod_{i=1}^dn_i $ observations in the first scale rectangle $[1, \lambda_1)\times \ldots \times [1, \lambda_d)$.
In general we consider d-dimensional $(k_1,k_2, \ldots, k_d)$ scale rectangle as
\begin{equation}
[\lambda_1^{k_1-1}, \lambda_1^{k_1})\times [\lambda_2^{k_2-1}, \lambda_2^{k_2})\times\ldots \times [\lambda_d^{k_d-1}, \lambda_d^{k_d}).
\end{equation}
The first scale rectangle is considered as the d-dimensional $(1,1,\ldots, 1)$ scale rectangle.
The MSI field with Markov property is called Markov MSI (MMSI). We show that the covariance function  of the sampled MMSI field  are presented by the covariance function of corresponding samples inside the first scale rectangles.
We present some proper estimation method based on this component-wise sampling scheme by extending the method of estimation the parameters for
DSI processes in \cite{R-M3, R-M4}.\\
This paper is motivated by applications in environmental and climate phenomena.
Precipitation is one of the key terms for balancing the energy budget, and one of the most challenging aspects of climate modeling.
Basic research performed in the statistical analysis and studied the variability in the distribution of rainfall to obtain accurate prediction
\cite{Olat}, \cite{Tul}.
As an example of MSI field, real data of the precipitation in some part of Brisbane area of Australia for some special period of time are considered.
The MSI behavior of these precipitation in three dimension as latitude, longitude and time are verified \cite{Australia}.
Also the corresponding time dependent scale  and Hurst parameter of the MSI field are estimated.
By estimating these parameters, we predict the precipitation in surface and time. All prediction methods have errors in predicting. So we use mean absolute percentage error (MAPE) as a statistical measure that calculate the error of the predictions.  We show that our predictions are  highly accurate. \\
\\
The rest of the paper is organized as follows. Section 2 is devoted to some preliminaries and also definitions of MSS and MSI fields.
The component-wise geometric  sampling scheme and some proper quasi-Lamperti transformation are defined in this section.
The definition of a two-dimensional simple fractional Brownian sheet (sfBs) as an example of MSI field is given in section 2 as well.
The characterization of covariance and spectral density functions of the two-dimensional scale invariant wide-sense Markov fields are presented
in section {3. In section 4 }we introduce a heuristic method for the estimation of Hurst parameter of MSI fields. Implying the rainfall data of
Brisbane area of Australian bureau of meteorology, their MSI property of the field is verified and also the scale and Hurst parameters
of this field are estimated. {\color{black}{\color{red}}Finally, we study the prediction of the precipitation and employ the mean absolute percentage  error(MAPE) index to determine the accuracy of the prediction.}

\renewcommand{\theequation}{\arabic{section}.\arabic{equation}}
\section{Theoretical Structure} \setcounter{equation}{0}
In this section, we present the definitions of the multi-selfsimilar (MSS) and multi-scale invariant (MSI) fields to be prescribed by some parameter space. Then we introduce a modified version of Lamperti transformation which provides a one to one correspondence between sampled MSS and discrete time stationary fields and also between sampled MSI and discrete time periodic fields respectively.\\
First we present the definition of periodic field and we use them as the Lamperti counterpart of self-similar field to obtain the  covariance structure  of MSI field.

\begin{defn}
The random field is said to be stationary field if for any  $\tau \in \mathbb{R}^d$
\begin{equation*}
\{\mathcal{S}_{\boldsymbol{\tau}}X(\mathbf{t}), t \in \mathbb{R}^d\}\overset{\mathcal{L}}{=}\{X(\mathbf{t}), t \in \mathbb{R}^d\},
\end{equation*}
where for any  $t$ and  $\boldsymbol{\tau}\in\mathbb{R}^d $, the shift operator $\mathcal{S}_{\boldsymbol{\tau}}$ acts as $\mathcal{S}_{\boldsymbol{\tau}}X(\mathbf{t}):=X(\mathbf{t}+{\boldsymbol{\tau}}).$
The random field is called  periodic with period $\boldsymbol{\tau}_0$ if the above equality holds  just for $\boldsymbol{\tau}=\boldsymbol{\tau}_0$.
\end{defn}

\begin{defn}
A second order random field is called periodically correlated (PC) if its mean and covariance function has a periodic structure for some $\boldsymbol{\tau}$, see \cite{Hurd}
$$E[X(\mathbf{t}+\boldsymbol{\tau})]=E[X(\mathbf{t})],\hspace{1in}
Cov(X(\mathbf{t}), X(\mathbf{s}))=Cov(X(\mathbf{t}+\boldsymbol{\tau}), X(\mathbf{s}+\boldsymbol{\tau})).$$
Periodic field with finite second moment is also a PC random field.
\end{defn}
\noindent It should be noted that a second order random field is square integrable over the parameter space.
Extending some definitions in Modarresi et al. \cite{R-M1} for DSI process with some parameter space, we present the following definitions.
\begin{defn}
A random field $\{X(\mathbf{k}), \mathbf{k}\in\check{\mathbf{T}}\}$ is called MSS with parameter space $\check{\mathbf{T}}$, where $\check{\mathbf{T}}$ is any subset of $[1, \infty)^d$ and for any
$\mathbf{k}_1=(k_{11}, k_{12},\ldots, k_{1d})^T$, $\mathbf{k}_2=(k_{21}, k_{22},\ldots, k_{2d})^T\in\check{\mathbf{T}}$
\vspace{-2mm}
\begin{equation}\label{eq15}
\{X(\mathbf{k}_2)\}\overset{\mathcal{L}}{=}\Big(\prod_{i=1}^d\big(\frac{k_{2i}}{k_{1i}}\big)^{H_i}\Big)\{X(\mathbf{k}_1)\}.
\end{equation}
The random field $X(\cdot)$ is called MSI with parameter space $\check{\mathbf{T}}$ and scale
$\boldsymbol{\Lambda}=(\lambda_1,\ldots,\lambda_d)$ if for any $\mathbf{k}_1,\mathbf{k}_2\in\check{\mathbf{T}}$, (\ref{eq15}) holds
where $k_{2i}=\lambda_ik_{1i}$ and $\lambda_i>1$ for $ i=1,\ldots,d$. Furthermore, it is to mention that the Hurst parameter in these fields
are not restricted with one and might be some other finite values.
\end{defn}
\noindent
Now, we are to consider some geometric sampling of the MSI field at points $\boldsymbol{\alpha}=(\alpha_1, \ldots, \alpha_d)^T$ where $\alpha_i>1$
for $i=1,\ldots ,d$.
\begin{rem}
By assuming $k_1, \ldots, k_d$ to be fixed integers $ k_i\in\{1, 2, \ldots, n_i\}$ and sampling of the MSI field at points
$\big \{\boldsymbol{\alpha}^{\mathbf{l}\mathbf{n}+\mathbf{k}}=(\alpha_1^{l_1n_1+k_1}
,\ldots, \alpha_d^{l_dn_d+k_d})^T: l_i\in \mathbb{N}_0, i=1,\ldots ,d\big \}$   we have
an MSS field with parameter space $\check{\mathbf{T}}=\{\boldsymbol{\alpha}^{\mathbf{l}\mathbf{n}+\mathbf{k}}, \mathbf{l}\in\mathbb{N}_0^d\}$.
\end{rem}
\noindent
Similar to the concept of the wide-sense self-similar process presented by Nuzman et al. \cite{Nuzman}, we have the following definition.
\begin{defn}
A second order random field $\{X(\mathbf{t}), \mathbf{t}\in\mathbb{R}^d_+\}$ is said to be wide-sense MSS
{with index  $\mathbf{H}=(H_1, \ldots, H_d)$}, if the following properties are
satisfied for
${\mathbf{t},\mathbf{t}_1,\mathbf{t}_2}
 \in \mathbb{R}^d_+ $ and $\mathbf{a}=  (a_1,\ldots,a_d)$ where $a_i>0$
\\
\hspace{1cm} $(i) E[X^2(\mathbf{t})]<\infty$
\\
$(ii) E[X(\mathbf{a\circ t})]=\big(\prod_{i=1}^d a_i^{H_i}\big)E[X(\mathbf{t})]$
\\
$(iii) E[X(\mathbf{a\circ t_1})X(\mathbf{a\circ t_2})]=\big(\prod_{i=1}^d a_i^{2H_i}\big)E[X(\mathbf{t_1})X(\mathbf{t_2})]$
\\
{where $\circ$ is the Hadamard product defined in (1.1)}. This field is called wide-sense MSI of index $\mathbf{H}$ and scale $\mathbf{a}^{\prime}=(a^{\prime}_1,\ldots,a^{\prime}_d)$
where $a^{\prime}_i>0$, if the above conditions hold for some $\mathbf{a}=\mathbf{a}^{\prime}$.
\end{defn}

\noindent
To find a one-to-one correspondence between the shift and renormalized operators and also between MSI and periodic fields, we introduce the
quasi-Lamperti transformation.
In the rest of the paper we consider MSS and MSI in the wide-sense fields, so for simplicity we omit the term "in the wide sense" henceforth.

\begin{defn}
The quasi-Lamperti transform $\mathcal{L}_{\mathbf{H},\boldsymbol{\alpha}}$   with positive Hurst vector $\mathbf{H}=(H_1,\ldots, H_d)$ and positive scale vector  $\boldsymbol{\alpha}=(\alpha_1, \ldots, \alpha_d)$,  operates on a random field $\{Y(\mathbf{t}), \mathbf{t}\in\mathbb{R}^d_+\}$ as
\vspace{-3mm}
\begin{equation}\label{eq16}
\mathcal{L}_{\mathbf{H},\boldsymbol{\alpha}}Y(\mathbf{t})=\big(\prod_{i=1}^d t_i^{H_i}\big)Y(\mathbf{Log}_{\boldsymbol{\alpha}}\mathbf{t}),
\end{equation}
where $\mathbf{Log}_{\boldsymbol{\alpha}}\mathbf{t}=(\log_{\alpha_1}t_1, \ldots,\log_{\alpha_d}t_d )^T$.
The corresponding inverse quasi-Lamperti transformation $\mathcal{L}^{-1}_{\mathbf{H},\boldsymbol{\alpha}}$ acts as
\vspace{-4mm}
\begin{equation}\label{eq17}
\mathcal{L}^{-1}_{\mathbf{H},\boldsymbol{\alpha}}X(\mathbf{t})=\prod_{i=1}^d \alpha_i^{-t_iH_i}X(\boldsymbol{\alpha}^{\mathbf{t}}),
\end{equation}
where $\boldsymbol{\alpha}^{\mathbf{t}}=(\alpha_1^{t_1},\ldots, \alpha_d^{t_d})^T$.
\end{defn}
\noindent
One can easily verify that $\mathcal{L}_{\mathbf{H},\boldsymbol{\alpha}}\mathcal{L}^{-1}_{\mathbf{H},\boldsymbol{\alpha}}X(\mathbf{t})= X(\mathbf{t})$ and
$\mathcal{L}^{-1}_{\mathbf{H},\boldsymbol{\alpha}}\mathcal{L}_{\mathbf{H},\boldsymbol{\alpha}}Y(\mathbf{t})= Y(\mathbf{t})$. If $\boldsymbol{\alpha}=(e,\ldots,e)^T$,
then $\mathcal{L}_{\mathbf{H},\boldsymbol{\alpha}}$ turn to be the usual Lamperti transformation, see \cite{Taqqu}.

\begin{prop}
The quasi-Lamperti transformation guarantees an equivalence between the shift operator $\mathcal{S}_{\mathbf{Log}_{\boldsymbol{\alpha}}\mathbf{\lambda }}$ and
the renormalized dilation operator $\mathcal{D}_{\mathbf{H},\mathbf{\lambda }}$ in the sense that, for any $\mathbf{\lambda}>\mathbf{0}$
\begin{equation}\label{eq18}
\mathcal{L}^{-1}_{\mathbf{H},\boldsymbol{\alpha}}\mathcal{D}_{\mathbf{H},\mathbf{\lambda }}\mathcal{L}_{\mathbf{H},\boldsymbol{\alpha}}
=\mathcal{S}_{\mathbf{Log}_{\boldsymbol{\alpha}}
\mathbf{\lambda}},
\end{equation}
where  $\mathcal{D}_{\mathbf{H},\mathbf{\lambda }}$    is defined by

\begin{equation*}
\mathcal{D}_{\mathbf{H},\boldsymbol{\Lambda}}X(\mathbf{t}):=\Big(\prod_{i=1}^d \lambda_i^{-H_i}\Big)X(\boldsymbol{\Lambda}\circ \mathbf{t}).
\end{equation*}

\end{prop}
\begin{proof}
By a similar method as in \cite{R-M1}, the validation of (\ref{eq18}) follows.
\end{proof}

\begin{cor}
If $\{X(\mathbf{t}), \mathbf{t}\in\mathbb{R}^d_+\}$ is a MSI field with scale ${\alpha}^{\mathbf{U}}$ then $\mathcal{L}^{-1}_{\mathbf{H},\boldsymbol{\alpha}}X(\mathbf{t})=Y(\mathbf{t})$ is periodic field with period $\mathbf{U}>\mathbf{0}$. Conversely if $\{Y(\mathbf{t}), \mathbf{t}\in\mathbb{R}^d\}$ is periodic field with period $\mathbf{U}$ then $\mathcal{L}_{\mathbf{H},\boldsymbol{\alpha}}Y(\mathbf{t})=X(\mathbf{t})$ is MSI with scale  ${\alpha}^{\mathbf{U}}$.
\end{cor}

\begin{rem}
If $X(.)$ is a MSS with parameter space $\check{\mathbf{T}}=\{(\alpha_1^{n_1U_1},\ldots, \alpha_d^{n_dU_d}), n_1,\ldots,n_d\in\mathbb{N}_0\}$ and Hurst vector $\mathbf{H}=(H_1,\ldots, H_d)^T$, then it is easy to show that its stationary counterpart $Y(.)$ has parameter space $\tilde{\mathbf{T}}=\{(n_1U_1,\ldots,n_dU_d), n_1,\ldots,n_d\in\mathbb{N}_0\}$.
\end{rem}
\noindent
A Brownian sheet is a natural extension of the Brownian motion to a two-dimensional random field and is one of the most important examples of the Gaussian random fields. Furthermore, some properties has been studied such as a method to study Brownian sheet by the linear stochastic partial differential equations \cite{Arat}.
Many data sets have anisotropic nature in the sense that they have different geometric and probabilistic characteristics along different directions, hence fractional Brownian motion is not adequate for modeling such phenomena. So several different classes of anisotropic Gaussian random fields such as fractional Brownian sheets have been introduced for theoretical and application purposes and some sample-function behavior of them studied \cite{Wang}, \cite{Xiao3}.
In the following, we present the definitions of centered Gaussian random field as the fractional Brownian sheet and the stationary rectangular increments property \cite{Mako}.
\begin{defn}
The normalized fractional Brownian sheet with Hurst index $\mathbf{ H}=\!(H_1, \ldots, H_n)$ {where $\mathbf{ H}\in\!\!(0,1)^n$}, is the centered Gaussian random field
$B^{\mathbf{ H}}=\{B^{\mathbf{ H}}(\mathbf{ t}), {\mathbf{ t}}\in{\mathbb{ R}}^n_{+}\}$ with covariance function
\begin{equation*}
E [B^{\mathbf{ H}}(\mathbf{ t})B^{\mathbf{ H}}(\mathbf {s})]=
2^{-n}\prod_{i=1}^{n}(|t_i|^{2H_i}+|s_i|^{2H_i}-|t_i-s_i|^{2H_i}),\,\,\,\,\,\,\,\, {\bf t}, {\bf s}\in{\Bbb R}^n_{+}.
\end{equation*}
This field is self-similar with index $\bf H$ by the definition in (1.1).
\end{defn}
\vspace{.4cm}

{\noindent
{\bf \large  Simple fractional Brownian sheet:}\\\

\noindent
Here we elaborate on flexible structure of MSI fields  in order to provide a platform for modeling  MSI field with  different correlation structure between  samples in different scale rectangles. So we define simple MSI field and in particular simple fractional Brownian sheet  { that defines a grid of scale rectangles where  despite the MSI behavior, samples inside each scale rectangle constitute some fractional Brownian sheet. To describe the structure of sfBs,  first we  define  a  double array sequence of fractional Brownian sheets 
where their  cross covariance functions follow component-wise discrete scale invariant property. }  


{
\noindent
\begin{defn} \label{d7}
Let 
$\big \{ B^{\mathbf{H}^{\prime}}_{n_1,n_2}(t_1,t_2) : n_1, n_2 \in {\Bbb N},\; (t_1, t_2)\in  [ \lambda_1^{n_1-1}, \lambda_1^{n_1})\times   [\lambda_2^{n_2-1}, \lambda_2^{n_2}) \big\}$  be a double array sequence of  fractional Brownian sheets with common Hurst indices $\mathbf{H}^{\prime}=(H_1^{\prime}, H_2^{\prime})$ 
where their cross covariance functions have component-wise discrete scale invariant property 
 with scales $\;\lambda_1,\lambda_2>1$ and same Hurst indices $H_1^{\prime}, H_2^{\prime}$.  
\end{defn}

\noindent
So   for fixed non-negative integers  $n_1, n_2$ , the  fractional Brownian sheet  $B^{\mathbf{H}'}_{n_1,n_2}(t_1, t_2)$  is defined on the scale rectangle 
${\mathbf A_{n_1,n_2}}:=[ \lambda_1^{n_1-1}, \lambda_1^{n_1})\times   [\lambda_2^{n_2-1}, \lambda_2^{n_2})$ 
and has the  covariance  structure defined by Definition 6. Also  the component-wise  discrete scale invariant property of the
     cross covariance  functions  states  that   for any 
 $(t_1, t_2)\in  {\mathbf A_{n_1,n_2}}   $, $\, (s_1, s_2)\in  {\mathbf A_{m_1,m_2}} $ and any non-negative integers $k_1,  k_2 $: }\vspace{-.25in}

\begin{eqnarray*}Cov\big(B^{\mathbf{H}^{\prime}}_{n_1+k_1, n_2+k_2}(\lambda_1^{k_1}t_1,\! \lambda_2^{k_2}t_2),\! B^{\mathbf{H}^{\prime}}_{m_1+k_1,m_2+k_2}(\lambda_1^{k_1}s_1,\! 
\lambda_2^{k_2}s_2)\big)
&\!\!\!=\!\!\!&\\ && \hspace{-1in}\prod_{i=1}^2 \! \lambda_i^{2k_iH^{\prime}_i}Cov\big(B^{\mathbf{H}^{\prime}}_{n_1,n_2}(t_1,\! t_2),\! B^{\mathbf{H}^{\prime}}_{m_1,m_2}(s_1,\!s_2)\big).\end{eqnarray*}

\noindent
Now we present the definition of simple fractional Brownian sheet (sfBs) as an example of MSI field. This Gaussian random field can be used to approximate many MSI fields.
\begin{defn}
The two-dimensional sfBs $\{X(\mathbf{t}), \mathbf{t}\in[1, \infty)^2\}$ is defined by
\begin{equation}\label{TS}
X(t_1,t_2)=\sum_{n_1=1}^{\infty}\sum_{n_2=1}^{\infty}\lambda_1^
{n_1(H_1-H_1^{\prime})}\lambda_2^{n_2(H_2-H_2^{\prime})}I_{[\lambda_1^{n_1-1},\lambda_1^{n_1})}(t_1)I_{[\lambda_2^{n_2-1},\lambda_2^{n_2})}(t_2)B^{\mathbf{H}^{\prime}}_{n_1,n_2}(t_1,t_2),
\end{equation}
is an  MSI field with Hurst $\mathbf{H}=(H_1, H_2)^T$ and scale $\boldsymbol{\lambda}=(\lambda_1, \lambda_2)$ where, $H_i>0 $, $\lambda_i>1$ for $i=1,2$,   $\big \{ B^{\mathbf{H}^{\prime}}_{n_1,n_2}(t_1,t_2) : n_1, n_2 \in {\Bbb N};\, t_1, t_2\in {\Bbb [0, \infty)}\big\}$ is introduced in  Definition \ref{d7}   and $I(.)$ is  an indicator function.   The intervals $[\lambda_1^{n_1-1},\lambda_1^{n_1})$ and $[\lambda_2^{n_2-1},\lambda_2^{n_2})$ are called $n_1$-th horizontal and $n_2$-th vertical scale intervals respectively. If $\mathbf{H}^{\prime}=(\frac{1}{2},\frac{1}{2})^T$, $B^{\mathbf{H}^{\prime}}_{n_1,n_2}(.,.)$  is the two-dimensional Brownian sheet and $X(.,.)$ is called two-dimensional simple Brownian sheet.
\end{defn}

\noindent
  One can easily verify the MSI property of sfBs through its covariance structure of corresponding sample:

$$
Cov\big(X(\lambda_1t_1,\lambda_2t_2), X(\lambda_1s_1,\lambda_2s_2)\big)
=\prod_{i=1}^2\lambda_i^{2H_i}Cov\big(X(t_1,t_2), X(s_1,s_2)\big).\hspace{3.5cm}\\
$$
Thus $X(t_1,t_2)$ is MSI field with scale parameters $\lambda_1$ and $\lambda_2$.
}

\renewcommand{\theequation}{\arabic{section}.\arabic{equation}}
\section{Two-dimensional Scale Invariant Markov Fields}
\setcounter{equation}{0}
The covariance function of a Markov random field $\{X(s,t), (s,t)\in\mathbb{R}^2\} $ is called separable if it satisfies
\begin{equation*}
Q(s_1,s_2, t_1,t_2)=Q_{1}(s_1,t_1)
Q_{2}(s_2,t_2),
\end{equation*}
where $Q(s_1,s_2, t_1,t_2)= Cov(X(s_1,s_2), X(t_1,t_2))$ and $Q_{1}(s_1,t_1)$, $Q_{2}(s_2,t_2)$ have the properties of the covariance functions of the
Markov processes, see \cite{Genton} and  \cite{Rosenfeld}.
Also there exists some statistical methods to test the separability of the covariance function of random fields, see \cite{fu}.\\
Let $\{X(s,t), (s,t)\in [1, \infty)^2\} $ be some MSI field with separable  covariance function. If the covariance function
$Q(s_1,s_2, t_1,t_2)$ has MSI property, then the covariance functions $Q_{1}(s_1,t_1)$ and $Q_{2}(s_2,t_2)$ can be considered as the covariance
functions of DSI Markov processes $X_1$ and $X_2$. The processes $X_1$ and $X_2$ exist as a result of the assumption that the field has DSI
property in each component in the introduced field.
So $Q_{1}(s_1,t_1)=Cov(X_1(s_1),X_1(t_1))$ and $Q_{2}(s_2,t_2)=Cov(X_2(s_2),X_2(t_2))$.
In this section we show that the covariance function of the MMSI field with separable property is characterized by the covariance functions of samples
on the first scale rectangle.\\
Following Remark 1, we consider sampled two-dimensional MMSI field as $\{X(\alpha_1^{n_1},\alpha_2^{n_2}), (n_1,n_2)\\\in\mathbb{Z}^2\} $ that
has separable covariance function with Hurst   $\mathbf{H}=(H_1, H_2)$ and scale  $\boldsymbol{\Lambda}= (\alpha_1^{T_1}, \alpha_2^{T_2})$.
Let
\begin{equation}\label{st1}
Q_{\mathbf{n}}^{\mathbf{H}}(\boldsymbol{\tau}):
=Cov[X(\alpha_1^{n_1+\tau_1},\alpha_2^{n_2+\tau_2}), X(\alpha_1^{n_1},\alpha_2^{n_2})]
\end{equation}
for
$\mathbf{n}=(n_1,n_2) , \boldsymbol{\tau}=(\tau_1, \tau_2)\in\mathbb{Z}^2$.  Also assume that  $\{X_i(\alpha_i^k), k\in\mathbb{Z}\}$ be a DSI
Markov process with parameters $(H_i,\alpha_i^{T_i})$ and covariance function
$Q_{i,n_i}^{H_i}(\tau_i)=Cov[X_i(\alpha_i^{n_i+\tau_i}),X_i(\alpha_i^{n_i})]$ for  $i=1,2$. So by the separable  property of the field we have that
\begin{equation*}
Q_{\mathbf{n}}^{\mathbf{H}}(\boldsymbol{\tau})=Q_{1,n_1}^{H_1}(\tau_1)Q_{2,n_2}^{H_2}(\tau_2).
\end{equation*}
Thus by Theorem 3.2 in \cite{R-M1},
\begin{equation}\label{r11}
Q_{\mathbf{n}}^{\mathbf{H}}(\mathbf{kT}+\boldsymbol{\nu})=[{\mathbf{h}}
(\boldsymbol{\alpha}^{\mathbf{T-1}})]^{\mathbf{k}} {\mathbf{h}}(\boldsymbol{\alpha}^{\boldsymbol{\nu}
+\mathbf{n-1}})[{\boldsymbol{h}}(\boldsymbol{\alpha}^{\mathbf{n-1}})]^{-1}Q_{\mathbf{n}}^{\mathbf{H}}(\mathbf{0}),
\end{equation}
and
\begin{equation*}
Q_{\mathbf{n}}^{\mathbf{H}}(-\mathbf{kT}+\boldsymbol{\nu})=\boldsymbol{\alpha}^
{-2\mathbf{kTH}}
Q_{\mathbf{n}+\boldsymbol{\nu}}^{\mathbf{H}}(\mathbf{kT}-\boldsymbol{\nu}),
\end{equation*}
where
$$
{\boldsymbol{h}}(\boldsymbol{\alpha}^{\mathbf{r}})={h}_1(\alpha_1^{r_1})
{h}_2(\alpha_2^{r_2}),\:\:\:\:\:\:\:\: \boldsymbol{\alpha}^
{-2\mathbf{kTH}}=\alpha_1^{-2k_1T_1H_1}\alpha_2^{-2k_2T_2H_2}
$$
and for $i=1,2,$ $k_i\in\{0,1,2,\ldots \}, \nu_i=0,1,\ldots, T_i-1$ and
$
{h}_i(\alpha_i^{r_i})
=\prod_{j=0}^{r_i}\frac{Q_{i,j}^{H_i}(1)}
{Q_{i,j}^{H_i}(0)}
$,
 ${h}_i(\alpha_i^{-1})=1$.
Furthermore, for $i=1,2$
$$\frac{Q_{i,j}^{H_i}(1)}{Q_{i,j}^{H_i}(0)}=\frac{Q_{i,j+T_i}^{H_i}(1)}{Q_{i,j+T_i}^{H_i}(0)},$$
so ${h}_i(\alpha_i^{lT_i+m-1})=\big({h}_i(\alpha_i^{T_i-1})\big)^l{h}_i(\alpha_i^{m-1})$.
This cause that the term ${\mathbf{h}}(\boldsymbol{\alpha}^{\boldsymbol{\nu}+\mathbf{n-1}})[{\boldsymbol{h}}(\boldsymbol{\alpha}^{\mathbf{n-1}})]^{-1}$
while $n_1+\nu_1-1>T_1$ or $n_2+\nu_2-1>T_2$ can be evaluated by the covariance and variance of the samples in the first scale interval.
Thus we have the following result.

\begin{prop}
Let $\{X(\alpha_1^{n_1},\alpha_2^{n_2}), (n_1,n_2)\in\mathbb{N}_0^2\}$
be a MMSI field with separable covariance function, Hurst parameter $\mathbf{H}=(H_1, H_2)$ and scale
$\boldsymbol{\Lambda}= (\alpha_1^{T_1}, \alpha_2^{T_2})$.
Then the covariance function of the field is characterized by the variance and covariance function of
samples in the first scale rectangle as shown by (\ref{r11}).
\end{prop}

\begin{rem}
The spectral density of MSI fields and sampled MMSI fields are characterized by the variance and covariance functions of the samples in
the first scale rectangle.
\end{rem}
\vspace{.4cm}

\noindent
{\bf \large  Multivariate self-similar Markov field:}\\\

\noindent
One of the main privileges of our method, which reveals by the proposed geometric sampling scheme is that in each dimension, any vertical and
horizontal strips on rectangles have DSI process that corresponds to multivariate self-similar process. Such correspondence traces its root back
to the work of Rozanov \cite{Rozanov} where the correspondence between PC process and the related multivariate stationary process are introduced.
As an example for the latter, the accumulated precipitation in an area in successive months can be considered as a PC process while the corresponding
multivariate stationary process can be considered as accumulated precipitation in successive Januaries and successive Februaries and so on,
that are stationary processes and have stationary cross correlations as well.
Now the MSI filed by the proposed geometric sampling method in two-dimensional case can provide some grid scale rectangles with $q_1q_2$ samples in
each rectangle, where $\lambda_1= \alpha_1^{q_1},\; \lambda_2= \alpha_2^{q_2}$.
Each point  as $X(i,j)$ in a scale rectangle  has corresponding  points as  $X(i+k_1{q_1}, j+k_2{q_2})$ in other scale rectangles   for all $k_1,k_2 \in\mathbb{N}$ that together provide a self-similar field corresponding to $X(i,j)$.
For $i=1, \ldots, q_1$ and $j=1, \ldots, q_2$ it provides an MSS field.\\
In one-dimensional case Modarresi et al. \cite{R-M1} explained the correspondence between DSI and multi-dimensional self-similar process, so by
the same manner, every MSI field corresponds to some multivariate MSS. Moreover, the MMSI field
$\{X(\alpha_1^{n_1},\alpha_2^{n_2}), (n_1,n_2)\in\mathbb{N}_0^2\}$ with scale  $\boldsymbol{\Lambda}= (\alpha_1^{T_1}, \alpha_2^{T_2})$,
corresponds to the $T_1T_2$-variate self-similar Markov field defined as $\big( Y_{0,0}(n_1,n_2),\ldots ,Y_{T_1-1, T_2-1}(n_1,n_2) \big)$ where
\begin{equation}
Y_{k_1,k_2}(n_1,n_2):=X(\alpha_1^{n_1T_1+k_1} , \alpha_2^{n_2T_2+k_2}),\label{eq140r}
\end{equation}

\noindent
$k_1=0,\ldots T_1 -1$, $k_2=0,\ldots T_2-1$. Hence
$$R^{\mathbf{H}}_{(k_1,k_2),(j_1,j_2)}((n_1,n_2),(m_1,m_2))=Cov(Y_{k_1,k_2}(n_1,n_2), Y_{j_1,j_2}(m_1,m_2))$$
{
is the cross covariance function that  can be written as

{\small
\begin{eqnarray}
R^{\mathbf{H}}_{(k_1,k_2),\!(j_1,j_2)}((n_1,n_2),(n_1\!+\!\tau_1,n_2\!+\!\tau_2))\!\!\!\!\!&=&\!\!\!\!\! \nonumber Cov(Y_{k_1,k_2}(n_1,n_2), Y_{j_1,j_2}(n_1+\tau_1,n_2+\tau_2))\\
\hspace{-.5cm}\!\!\!\!\!&=&\!\!\!\!\!
\nonumber
Cov[X(\alpha_1^{n_1T_1+k_1} , \alpha_2^{n_2T_2+k_2}),X(\alpha_1^{(n_1+\tau_1)T_1+j_1} , \alpha_2^{(n_2+\tau_2)T_2+j_2})]\\
\!\!\!\!\!&=&\!\!\!\!\!
\nonumber
{\alpha_1}^
{2n_1T_1H_1}{\alpha_2}^
{2n_2T_2H_2}Cov[X(\alpha_1^{\tau_1T_1+j_1} , \alpha_2^{\tau_2T_2+j_2}),X(\alpha_1^{k_1} , \alpha_2^{k_2})]\\
\!\!\!\!\!&=&\!\!\!\!\!
\nonumber
\boldsymbol{\alpha}^
{2\mathbf{nTH}}
{Q}_{\mathbf{k}}^{\mathbf{H}}(\boldsymbol{\tau}\mathbf{T}+\mathbf{j-k})\\
\!\!\!\!\!&=&\!\!\!\!\!
\boldsymbol{\alpha}^
{2\mathbf{nTH}}[{\mathbf{h}}
(\boldsymbol{\alpha}^{\mathbf{T-1}})]^{\boldsymbol{\tau}} {\mathbf{h}}(\boldsymbol{\alpha}^{\mathbf{j-1}})[{\boldsymbol{h}}(\boldsymbol{\alpha}^{\mathbf{k-1}})]^{-1}
{Q}_{\mathbf{k}}^{\mathbf{H}}(\mathbf{0}),
\end{eqnarray}
}}
where $\boldsymbol{\alpha}^
{2\mathbf{nTH}}={\alpha_1}^
{2n_1T_1H_1}{\alpha_2}^
{2n_2T_2H_2}$.
{
We remind that the validity of the third equality follows from Definition 4, the fourth equality from (3.1) and the last equality from (3.2).
}
So we have the following proposition.

\begin{prop}
Let $\big( Y_{0,0}(n_1,n_2),\ldots ,Y_{T_1-1, T_2-1}(n_1,n_2) \big)$, ${(n_1,n_2)\in \mathbb{N}_0^2}$ be the multivariate self-similar Markov
field defined by (3.3). Then its cross covariance function is characterized by (3.4).
\end{prop}
\vspace{.4cm}

{\noindent
{\bf \large  Scale Markov Property:}\\\
\newcommand*\vtick{\textsc{\char13}}

\noindent
 The Markov property can follows by some sub-sequences of a sequence of random variables which itself has not Markov property.
 Examples of this can be described as the sub-sequences of some PC processes.  The accumulated precipitation on the same month of successive years in some specific place and also the  traffic volume of a high-way at some specific hour  of each working day are such examples, { see \cite{aiye},  \cite{cance}, \cite{tian}}. One may call such Markov property as periodic Markov property where subsequences are obtained at points $\big\{t+n\tau, n\in {\Bbb N} \big \}$  for any fixed $t$, where $\tau$ is the period of the main PC sequences.    In contrast we study the scale Markov property that the subsequences of  a DSI  sequence at points $\big \{\lambda^n t, n \in {\Bbb N} \big \}$ for any fixed $t$,  where $\lambda$ is the scale of DSI  sequence, have Markov property.  We  describe this as component-wise scale Markov property for two dimensional MSI  field  by   the  followings.
{

\begin{defn} \label{sta}
Let $\big \{X(t_1,t_2), (t_1,t_2) \in {\Bbb  R}^2_+\big \}$ be the two-dimensional MSI field with scales  $(\lambda_1, \lambda_2)$  introduced by Definition 4.   This MSI field is said to have component-wise scale Markov property 
at points $\big \{(\lambda_1^{n_1}t_1, \lambda_2^{n_2}t_2 ), n_1, n_2\in{\Bbb N^0}\big\}$ for fixed $(t_1,t_2)\in {\Bbb R}^2_+,$  where ${\Bbb N}^0$ is the set of non-negative integers,
if the self-similar processes   
 $\big \{X_1(\lambda_1^{n_1}):=X(\lambda_1^{n_1}t_1, t_2), n_1 \! \in {\Bbb N^0}\big \}$ and $\big \{ X_2(\lambda_2^{n_2}):=X(t_1, \lambda_2{^{n_2}} t_2), n_2 \in {\Bbb N}^0\big \}$    have Markov property. 
\end{defn}

\noindent 
Let $X(t_1, t_2)$ be ths sfBs introduced by Definition 8. By the following remarks,  we present some characterization method for such scale Markov property.  

\begin{rem}
For fixed $(t_1,t_2)\in {\Bbb R}^+$, let  $\big \{X_i(\lambda_i^{n_i}), n_i \in \mathbb{N}\big \};$  $i=1,2$ be the self-similar processes defined  in   Definition 9.  Following the  method of Modarresi and Rezakhah \cite{ R-M2, R-M5}, we  assume that the subsidiary self-similar processes $\big \{ X_i^*(n_i)=X_i(\lambda_i^{n_i}), n_i \in \mathbb{N}\big \};$  $i=1,2$,  follow some  AR(1) models  as 
$$X_i^*(n_i)=\theta_i  X_i^*(n_i-1)+Z_i^*(n_i), \;\;\;\;\;\;\;\;\; Z^*_i(n_i)= \lambda_i^{n_iH_i}Z_i(n_i);   \;\;\;\;\;i=1,2,$$
where    $\big \{Z_i(n_i), n_i\in {\Bbb N} \big \}$'s   are Gaussian    white noises {\color{black} and $\theta_i$'s  are constants. 
 Thus  
their quasi Lamperti  ${\cal L}_{H_i, \lambda_i}$  counter parts $Y_i(n_i)= \lambda_i^{-n_iH_i}X_i(\lambda_i ^{n_i})$ are  stationary processes and  follow AR(1) models 
$Y_i(n_i)=\lambda_i^{-H_i}\theta_i Y_i(n_i-1) +Z_i(n_i),\;  i=1,2$.}
\end{rem}

\begin{rem}
Let  $\big\{Y_1(1), Y_1(2), \ldots ,Y_1(N_1)\big\}$ and $\big \{ Y_2(1), Y_2(2),\ldots ,Y_2(N_2)\big\}$ be samples of the stationary processes $Y_1(\cdot)$ and $Y_2(\cdot)$ introduced in Remark 4.  Also assume that  $R_{Y_1}(k)$ and $R_{Y_2}(k)$ denote the  corresponding sample auto-correlation functions at lag $k$. 
By  the results of  \cite[Section~8.2]{Broc}  and \cite[Section~4.2]{Chat}, the stationary processes $Y_1(\cdot)$ and $Y_2(\cdot)$  are accepted to follow AR(1) models   if the corresponding  lag 2 sample partial autocorrelation functions  $\alpha_{Y_i} (2)=(R_{Y_i}(2)-R_{Y_i}^2(1))/(1-R_{Y_i}^2(1))$
lie in the interval $\big(-1.96/(\sqrt{N_i}), 1.96/(\sqrt{N_i})\big)$ for $i=1,2$ 
respectively.   So they have Markov property and by Remark 4  the 
processes $\big \{ X_1^*(n_1), n_1 \in{\Bbb N}\big \}$ and
$\big \{X_2^*(n_2), n_2 \in {\Bbb N} \big \}$  follow AR(1) models and have Markov property as well. Thus  by Definition 9 {\color{black}  the  sfBs $X(t_1,t_2)$ has component-wise  scale Markov property at points $\big \{(\lambda_1^{n_1}t_1, \lambda_2^{n_2} t_2 ), n_1\in{\Bbb N^0}\big\}$ for fixed $(t_1,t_2)\in {\Bbb R}^2_+$.}
\end{rem}

\begin{rem}
{\color{black}Using  Remark 4, one   can easily verify that the autocorrelation functions of the processes  $\big \{ Y_1(n_1), n_1 =1,2,\ldots ,N_1 \big \}$  and  $\big \{Y_2(n_2), n_2 =1,2,\ldots ,N_2 \big \}$ at all lags are equal to the autocorrelation functions of the processes  $\big \{ X_1^*(n_1)=X_1(\lambda_1^{n_1}), n_1 =1,2,\ldots ,N_1\big \}$ and
$\big \{X_2^*(n_2)=X_2(\lambda_2^{n_2}), n_2=1,2,\ldots ,N_2 \big \}$ at the same lag respectively, say  $ R_{Y_i}(k)=R_{X_i^*}(k) $\, for all $k \in {\Bbb N}$ .} {\color{black} This equality is valid by the fact that  for $i=1,2, \;$  $Y_i(n_i)=\lambda_i^{-n_iH_i}X_i^*(n_i)$ and 
 $Y_i(n_i+k)=\lambda_i^{-(n_i+k)H_i}X_i^*(n_i+k)$. So
$\mbox{Var}\, (Y_i(n_i))= \lambda_i^{-2n_iH_i}\, \mbox{Var}\,  (X_i^*(n_i))$ and $\mbox{Var} \, (Y_i(n_i+k))= \lambda_i^{-2(n_i+k)H_i}\, Var (X_i^*(n_i+k))$ 
and $\mbox{Cov} \, (  Y_i(n_i) , Y_i(n_i+k)   )= \lambda_i^{-(2n_i+k)H_i} \,  \mbox{  Cov} \,(X^*_i(n_i), X_i^*(n_i+k)).$
 Thus $R_{Y_i}(k)=\mbox{Corr}\, (Y_i(n), Y_i(n+k)=\mbox{Corr} \, (X_i^*(n), X_i^*(n+k))=R_{X_i^*}(k)$. 
Therefore   the sfBs $X(t_1, t_2)$ is accepted to have 
  the component-wise scale Markov property  if   the lag 2 sample partial autocorrelations  $\; \alpha_{Y_i} (2)=(R_{X_i^*}(2)-R_{X^*_i}^2(1))/(1-R_{X^*_i}^2(1))$ lie in the intervals $\big(-1.96/(\sqrt{N_i}), 1.96/(\sqrt{N_i})\big)$ for $i=1,2$ respectively.}
\end{rem}}

{

\section{Real Data modeling}\setcounter{equation}{0}
In this section we study the precipitation data on a region of Brisbane  area in Australia for two days (25 and 26 January 2013).
This study follows by the evaluation of such precipitation on squares with side length 2km of a grids  over a 512 km $\times$ 512 km in this region.
Thus the precipitation values are considered as a $256\times256$ matrix.
The precipitation of rainfall in the area is depicted in Figure 1.
 The region was affected by extreme rainfall and subsequent flood.
This  rainfall  data is provided by the Australian bureau of meteorology \cite{Australia}.
	
\subsection{Estimation of Scale parameters of MSI field}
Here we describe the estimation method for fitting simple MSI filed associated with some scale rectangles as a discretized approximation of MSI field which has component-wise scale invariant property  and component-wise self similarity.
Then the estimation of scale  parameters  along the horizontal, vertical and time coordinates are followed for the precipitation on certain area in the described  region are followed.
This circumscription is specified by yellow rectangle in Figure 1.
The selected part, plotted in Figure 2, is a 120 km $\times$ 100 km area which correspondence to a $60\times50$ square areas with sides of  length 2km.
The MSI property of these precipitation is justified by detecting the scale invariant behavior for the accumulated precipitation on vertical and horizontal strips, and the precipitation for the whole are in successive 30 minutes of time on 25th and 26th January 2013.'

\begin{figure}[!tbp]
  \centering
  \begin{minipage}[b]{0.47\textwidth}
    \includegraphics[width=\textwidth]{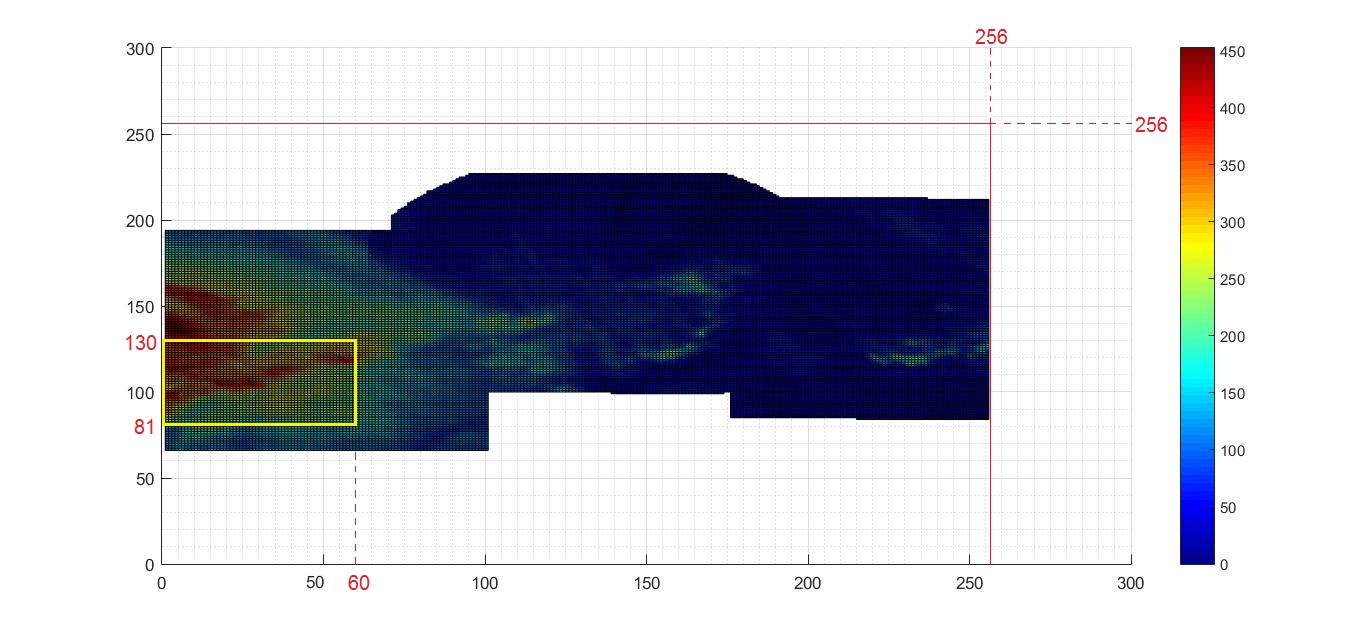}
    \caption{\footnotesize{Two-dimensional image of precipitation data (data unit is mm) over a 512 km $\times$ 512 km region  in the Brisbane
    area for two days (25th and  26th  January 2013) and the  selected area that specified by yellow rectangle.}}
  \end{minipage}
  \hfill
\vspace{-2.2in}

\hspace{3in}  \begin{minipage}[b]{0.44\textwidth}
    \includegraphics[width=\textwidth]{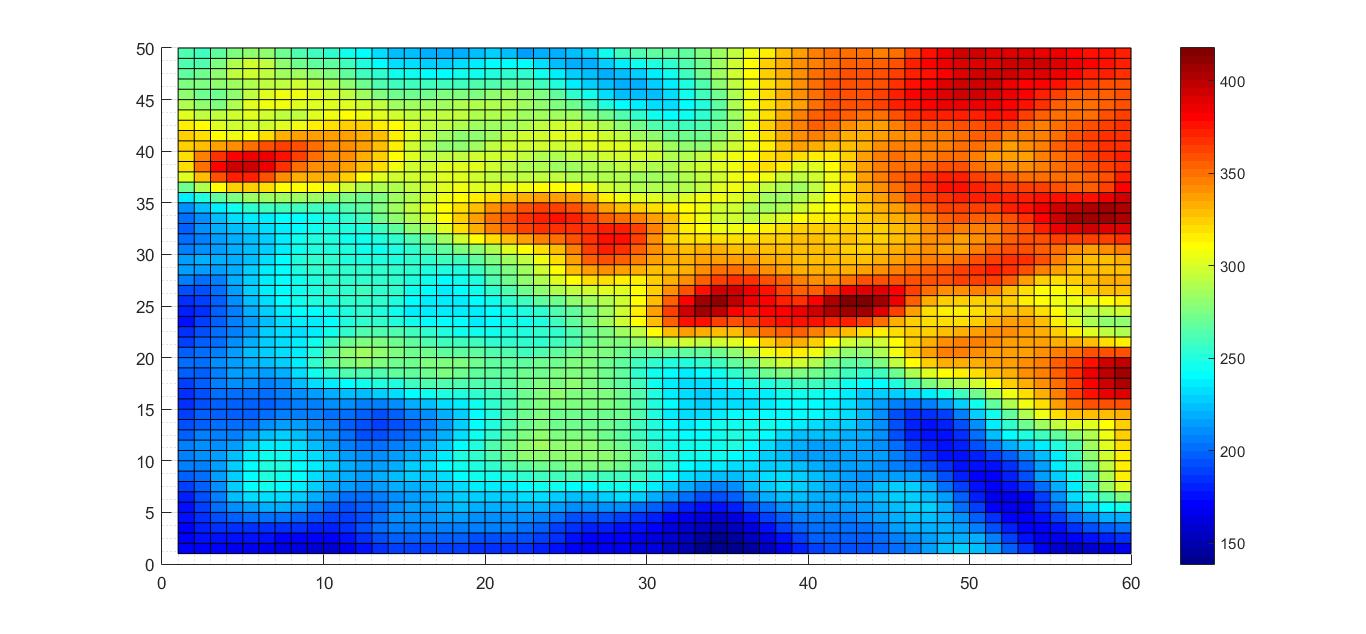}
    \caption{\footnotesize{Two-dimensional image of precipitation data for the selected area from Figure 1.}}\vspace{.4in}

  \end{minipage}
\end{figure}

\noindent Let $X_{ij}$ be the value of precipitation on $ij$-th square with vertices  in $(i,j),(i,j-1),(i-1,j),(i-1,j-1)$, where  $i=1,\ldots, 60$ and $j=1,\ldots, 50$. Also let  $X_{i.}=\sum_{k=1}^{50}X_{ik}$ denotes the accumulated precipitation
 on $i$-th vertical strip and $X_{.j}=\sum_{l=1}^{60}X_{lj}$  the accumulated precipitation  on $j$-th horizontal strip, where $i=1,\ldots, 60$ and $j=1,\ldots, 50$.
Table 1 and Table 2 show such amounts on successive vertical and horizontal strips respectively.  By plotting  accumulated precipitations in Figures 3 and 4,   the corresponding  scale intervals   are detected by fitting some proper parabolas that  their end points are shown with vertical red lines. These plots high-lights two characteristic  features of DSI processes as the ratio of the length of successive scale intervals are nearly the same which is called scale parameter and having somehow similar dilation in successive scale  intervals.
 In Figure 3 this method detects  three successive scale intervals for the accumulated precipitation on vertical strips  with  end points   $a_1=0$, $a_2=14$, $a_3=31$, $ a_4=52$, and in Figure 4 it detects three successive scale intervals  for the accumulated  precipitation on horizontal strips  with  end points $b_1=0$, $b_2=10$, $b_3=23$, $b_4=40$.
Following the  estimation method for scale parameter in \cite{R-M3}}, we evaluate the corresponding time varying scale parameters by $\lambda_{1,n-1}=\frac{a_{n+1}-a_n}{a_n-a_{n-1}}$ and $\lambda_{2,n-1}=\frac{b_{n+1}-b_n}{b_n-b_{n-1}}$ for $n=2,3$. This leads to find $\Lambda_1:=(\lambda_{1,1},\lambda_{1,2})=(1.214, 1.235)$ as the values of scale parameter for DSI process of accumulated precipitation on the vertical strips and $\Lambda_2:=(\lambda_{2,1},\lambda_{2,2})=(1.3, 1.307)$ on the horizontal strips.
{So we estimate  horizontal and vertical scale parameters as
\begin{equation}\label{s1}
{\hat\lambda_1}=\frac{{\lambda}_{1,1}+{\lambda}_{1,2}}{2}=\frac{1.214+1.235}{2}=1.224,\:\:\:\:\:\:\:\:\:
{\hat\lambda_2}=\frac{{\lambda}_{2,1}+{\lambda}_{2,2}}{2}=\frac{1.3+1.307}{2}=1.303.
\end{equation}}
\begin{table}[h!]
\begin{center}

\small
 \begin{tabular}{||l l l l l l l l l l||}

\hline
&	&	&	&	$X_{i.}$ & 	&	&	&	&	
\\
\hline
11169&	11448&	11812&	12174&	12454&	12620&	12673&	12673&	12663&	12636

\\
\hline
12590&	12545&	12516&	12499&	12511&	12567&	12692&	12855&	13013&	13201

\\
\hline
13406&	13561&	13646&	13694&	13750&	13802&	13813&	13754&	13613&	13460

\\
\hline
13382&	13409&	13532&	13696&	13896&	14138&	14377&	14560&	14672&	14730

\\
\hline

14765&	14788&	14820&	14876&	14956&	15051&	15152&	15235&	15299&	15365

\\
\hline

15433&	15496&	15572&	15687&	15850&	16044&	16276&	16536&	16787&	17009

\\
\hline
\end{tabular}
\caption{Sum of precipitation data on vertical strips as millimeters}
\label{table:1}
\end{center}
\end{table}

\begin{table}[h!]
\begin{center}
\small
\begin{tabular}{||l l l l l l l l l l||}
\hline
&	&	&	&	$X_{.j}$ &	&	&	&	&	
\\
\hline
10593&
10964&
11379&
11862&
12443&
13051&
13578&
13927&
14105&
14187
\\
\hline

14171&
14131&
14184&
14395&
14835&
15431&
16064&
16638&
17039&
17179
\\
\hline

17161&
17257&
17546&
17927&
18188&
18246&
18303&
18371&
18387&
18451
\\
\hline

18728&
19068&
19215&
19214&
19152&
19194&
19419&
19634&
19668&
19491
\\
\hline

19246&
19017&
18882&
18772&
18575&
18340&
18184&
18045&
17809&
17553
\\
\hline
\end{tabular}
\caption{ Sum of precipitation data on  horizontal strips as millimeters}
\label{table:2}
\end{center}
\end{table}

\begin{figure}[!tbp]
  \centering
  \begin{minipage}[b]{0.49\textwidth}
    \includegraphics[width=\textwidth]{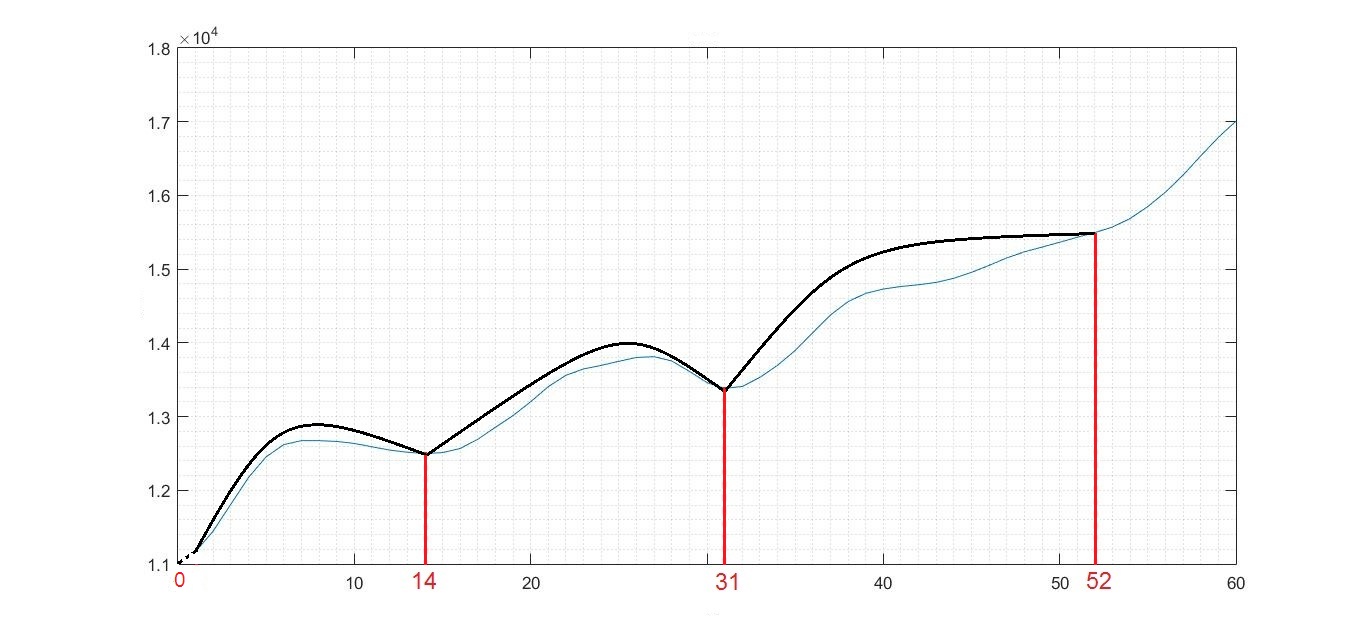}
    \caption{\footnotesize{Fitted black curves for  precipitation data on vertical  strips  and revealing the corresponding scale intervals by red lines.}}
  \end{minipage}
  \hfill
  \begin{minipage}[b]{0.5\textwidth}
    \includegraphics[width=\textwidth]{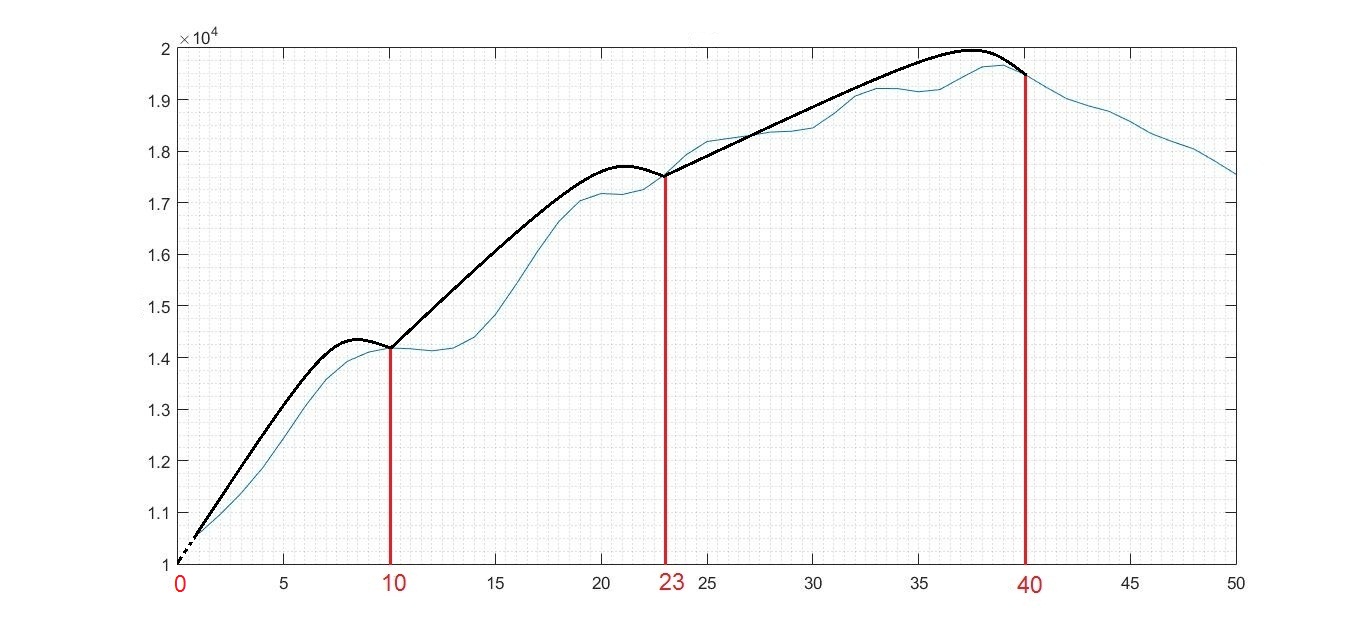}
    \caption{\footnotesize{Fitted black curves for  precipitation data on horizontal strips and revealing the corresponding scale intervals by red lines.}}
  \end{minipage}
\end{figure}
\noindent
{
For each horizontal and vertical scale interval we consider two equal subintervals which are indicated with green dashed lines in Figures 5, 6. For the horizontal scale intervals, we partition each subinterval into seven equally length interval in Figure 5, where the accumulated precipitation on the corresponding successive vertical strips of the $m$-th subinterval of the $n$-th horizontal scale interval  are denoted by $x_{(n,m)1}, \ldots , x_{(n,m)7}$.
So $x_{(n+1,m)k}\overset{\mathcal{L}}{=}\lambda_1^{\mathcal{H}_{(n,m)1}}x_{(n,m)k}$ where $\mathcal{H}_{(n,m)1}$ is corresponding  Hurst parameter
in between $n$-th and $(n+1)$-th  horizontal  scale intervals of sfBs (2.5)
 related to the $m$-th subintervals.
Also each subinterval in vertical scale intervals is divided into five equally length interval in Figure 6, where the accumulated precipitation on the corresponding successive horizontal strips of the $m$-th subinterval of the $n$-th vertical scale interval  are denoted by $y_{(n,m)1}, \ldots , y_{(n,m)5}$. So  $y_{(n+1,m)l}\overset{\mathcal{L}}{=}\lambda_2^{\mathcal{H}_{(n,m)2}}y_{(n,m)l}$ where $\mathcal{H}_{(n,m)2}$ is corresponding  Hurst parameter
 in between $n$-th and $(n+1)$-th vertical scale intervals of sfBs (2.5)
 related to the $m$-th subintervals
}

\subsection{Estimation of Hurst parameters}
 {
 For the estimation of the Hurst parameters
  of fitted sfBs (2.5)
  to  the precipitation area, using  the notations described at the end of subsection 4.1, we have that
   $$x_{(n+1,m)k}^2\overset{\mathcal{L}}{=}\lambda_1^{2\mathcal{H}_{(n,m)1}}x_{(n,m)k}^2
  $$
  for $k=1,\ldots , 7$ and
  $$y_{(n+1,m)l}^2\overset{\mathcal{L}}{=}\lambda_2^{2\mathcal{H}_{(n,m)2}}y_{(n,m)l}^2
  $$
  for $l=1,\ldots , 5$.
 Now we propose to estimate $\lambda_{i}^{\mathcal{H}_{(n,m)i}}$  by the ratio of the corresponding quadratic means for $i=1,2$.
Hence,
\begin{equation*}
\hat{\mathcal{H}}_{(n,m)i}=\frac{\log(SS_{(n+1,m),i}/SS_{(n,m),i})}{2\log\lambda_{i,m}},
\end{equation*}
where
\begin{equation}
SS_{(n,m),1}=\frac{1}{7}\sum_{k=1}^{7}x^2_{(n,m)k},
\:\:\:\:\:\:\:\:SS_{(n,m),2}=\frac{1}{5}\sum_{k=1}^{5}y^2_{(n,m)k}.
\end{equation}
}
\noindent
 {\color{black}
 {The values of $x_{(n,m)k}$  and $y_{(n,m)l}$ } are presented in Tables 3 and 4 by $k$ orders.
\begin{table}[h!]
\begin{center}
\small
\begin{tabular}{||c c c c c||}
\hline
n \!\!\!\!\!\!&\vline & \!\!\! \!\!\!$x_{(n,1)k}$   \!\!\!\!\!\! &\vline & \!\!\!\!\!\!$x_{(n,2)k}$  \\ [0.5ex]
 \hline\hline
1 \!\!\!\!\!\!&\vline &\scriptsize \!\!\!\!\!\!
11169 -	 11448 -	 11812 - 	12174 -	 12454 	- 12620 	-  12673

\!\!\!\!\!\!&\vline&\!\!\!\!\!\! \scriptsize 12673 -	12663 - 	12636	- 12590 -	12545 -12516 -	12499
\\
 \hline
2\!\!\! \!\!\!&\vline&\scriptsize \!\!\!\!\!\!15200 -	15310 -	 15513 -	15754  -	 16014 -	16319 	- 16519
\!\!\!\!\!\! &\vline&\!\!\!\!\!\!\scriptsize 16601 -	16676 -	 16753 - 	16748 	- 16617 -16406 -	 16289
 \\
 \hline
3\!\!\! \!\!\!&\vline& \!\!\!\!\!\! \scriptsize 20175 - 20462 -	20965 - 	21446 - 	21896 -	22066 - 	22159
 \!\!\!\!\!\!&\vline&\!\!\!\!\!\!\scriptsize   22214  -	22354 -	22529 -	22770  - 	22917 -	23082 -	 23213
\\
 \hline
\end{tabular}
\caption{precipitation values on partitions in subintervals  along vertical strips.}
\label{table:3}
\end{center}
\end{table}

\begin{table}[h!]
\begin{center}
\small
\begin{tabular}{||c c c c c||}
\hline
n &\vline &  $y_{(n,1)l}$    &\vline & $y_{(n,2)l}$  \\ [0.5ex]
 \hline\hline
1 &\vline &\scriptsize
10593 -
10964 -
11379 -
11862 -
12443

 &\vline& \scriptsize
13051 -
13578 -
13927 -
14105 -
14187

\\
 \hline
2 &\vline&\scriptsize
18410 -
18402 -
18629 -
19361 -
20377

 &\vline&\scriptsize
21342 -
22085 -
22326 -
22377 -
22723

 \\
 \hline
3 &\vline&\scriptsize
30659 -
31024 -
31192 -
31309 -
31952

 &\vline&\scriptsize
32592 -
32608 -
32761 -
33302 -
33259

\\
 \hline
\end{tabular}
\caption{ precipitation values on partitions in subintervals along horizontal strips.}
\label{table:4}
\end{center}
\end{table}
The ratios of the mean squares and estimated  Hurst values  are shown in Tables 5 and 6.
\begin{table}[h!]
\begin{center}
\small
 \begin{tabular}{||c c c c c c  ||}
 \hline
 n  &\vline & $\frac{SS_{(n+1,1),1}}{SS_{(n,1),1}}$& $\hat{\mathcal{H}}_{(n,1)1}$&  $\frac{SS_{(n+1,2),1}}{SS_{(n,2),1}}$& $\hat{\mathcal{H}}_{(n,2)1}$\\ [0.5ex]
 \hline\hline
 1 &\vline & 1.718&1.40& 1.736 &1.42\\
 \hline
2 &\vline& 1.819& 1.42  &1.878& 1.49 \\
\hline

\end{tabular}
\caption{ The ratios of squared of quadratic means and  corresponding  Hurst values  along vertical  strips.}
\label{table:5}
\end{center}
\end{table}

\begin{table}[h!]
\begin{center}
\small
 \begin{tabular}{||c c c c c c ||}
 \hline
 n  &\vline & $\frac{SS_{(n+1,1),2}}{SS_{(n,1),2}}$& $\hat{\mathcal{H}}_{(n,1)2}$&  $\frac{SS_{(n+1,2),2}}{SS_{(n,2),2}}$& $\hat{\mathcal{H}}_{(n,2)2}$\\ [0.5ex]
 \hline\hline
 1 &\vline & 2.76& 1.93& 2.60& 1.81\\
 \hline
 2 &\vline&  2.69 &1.85& 2.20 & 1.47  \\
 \hline

\end{tabular}
\caption{ The ratios of squared of quadratic means and corresponding Hurst values along horizontal strips.}
\label{table:6}
\end{center}
\end{table}
}

Hence,
\begin{figure}[!tbp]
  \centering
  \begin{minipage}[b]{0.485\textwidth}
    \includegraphics[width=\textwidth]{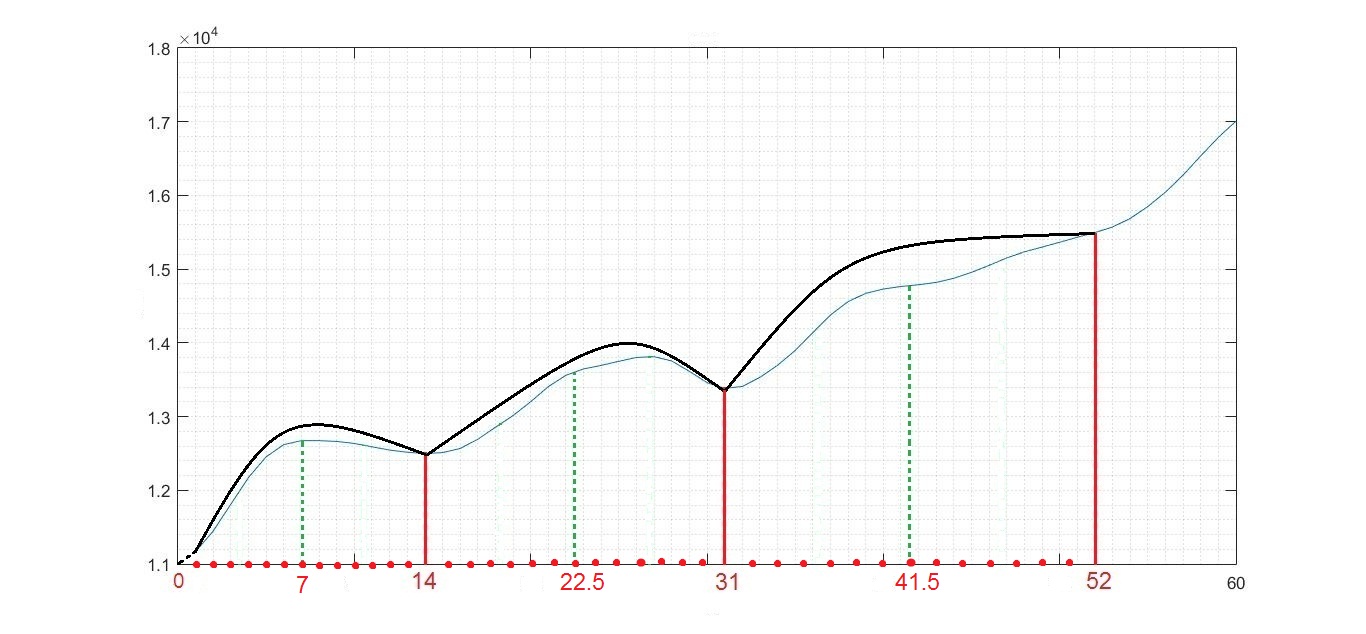}
    \caption{\footnotesize{Each horizontal scale interval is divided into two equally length  subintervals by green dashed lines.}}
  \end{minipage} 
  \hfill \hfill 
  \begin{minipage}[b]{0.495\textwidth}
    \includegraphics[width=\textwidth]{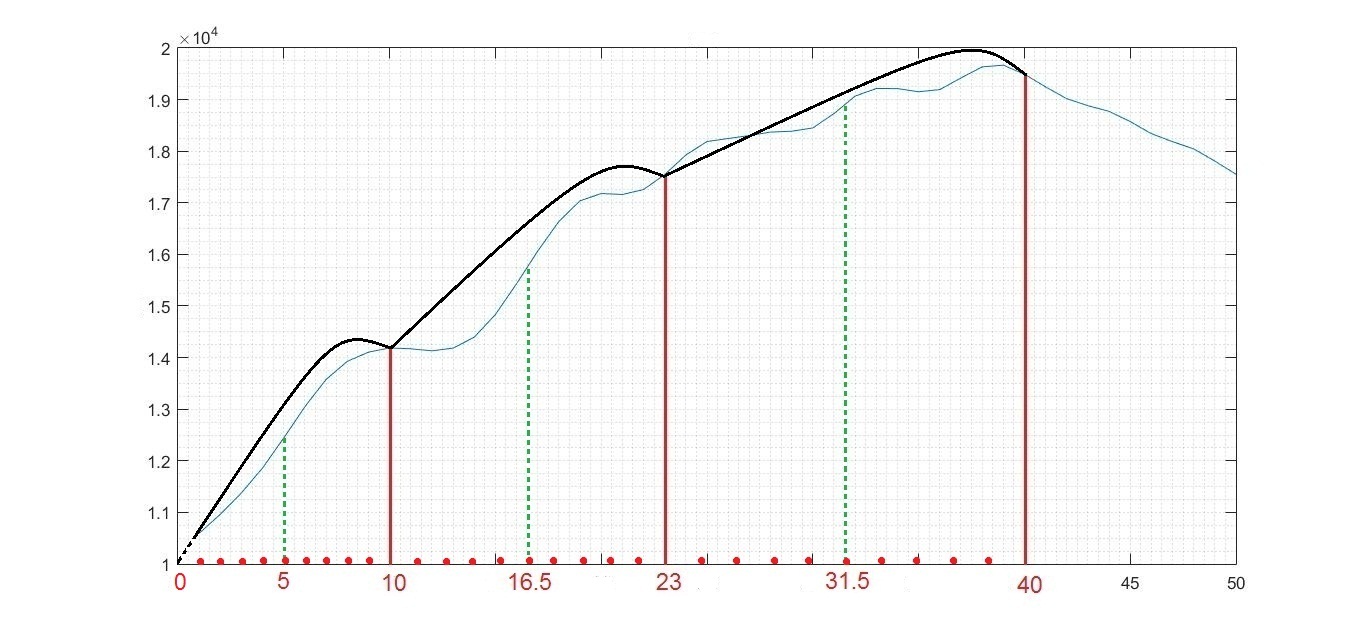}
    \caption{\footnotesize{Each vertical scale interval is divided into two equally length  subintervals by green dashed lines..}}
  \end{minipage}
\end{figure}

$$\hat{{\bf H}}_{1}:=(\hat{\mathcal{H}}_{(1,1)1},\hat{\mathcal{H}}_{(1,2)1}, \hat{\mathcal{H}}_{(2,1)1},\hat{\mathcal{H}}_{(2,2)1})=(1.40, 1.42, 1.42, 1.49)$$
and
$$\hat{{\bf H}}_{2}:=(\hat{\mathcal{H}}_{(1,1)2},\hat{\mathcal{H}}_{(1,2)2}, \hat{\mathcal{H}}_{(2,1)2},\hat{\mathcal{H}}_{(2,2)2})=(1.93, 1.81, 1.85, 1.47).$$
{
Denoting the  Hurst parameter
in between $n$-th and $(n+1)$-th  horizontal  scale intervals by $\hat{{{H}}}_{1,n}$ and
 for vertical scale intervals by $\hat{{{H}}}_{2,n}$, we have hat
}
\begin{equation*}
\hat{{{H}}}_{1,1}=\frac{\hat{\mathcal{H}}_{(1,1)1}+\hat{\mathcal{H}}_{(1,2)1}}{2}=1.41,\:\:\:\:\:\:
\hat{{{H}}}_{1,2}=\frac{\hat{\mathcal{H}}_{(2,1)1}+\hat{\mathcal{H}}_{(2,2)1}}{2}=1.46,
\end{equation*}
\begin{equation*}\hat{{{H}}}_{2,1}=\frac{\hat{\mathcal{H}}_{(1,1)2}+\hat{\mathcal{H}}_{(1,2)2}}{2}=1.87,\:\:\:\:\:\:\:
\hat{{{H}}}_{2,2}=\frac{\hat{\mathcal{H}}_{(2,1)2}+\hat{\mathcal{H}}_{(2,2)2}}{2}=1.66.
\end{equation*}
{
So we estimate  horizontal and vertical Hurst parameters
  of fitted sfBs (2.5)
 by
\begin{equation}\label{s2}
\small
\hat{{{H}}}_{1}=\frac{\hat{{{H}}}_{1,1}+\hat{{{H}}}_{1,2}}{2}=\frac{1.41+1.46}{2}=1.435\:\:\:\:\:\:\:
\hat{{{H}}}_{2}=\frac{\hat{{{H}}}_{2,1}+\hat{{{H}}}_{2,2}}{2}=\frac{1.87+1.66}{2}=1.765
\end{equation}
}
\noindent
Now we show the DSI behavior of the  precipitation with respect to time in the selected whole area which is shown in Figure 2.
For this, we consider the accumulated precipitation in this area for every 30 minutes on 25th and 26th January 2013 which are depicted in
{\color{black}
Table 7
}by their orders of time.

\begin{table}[h!]
\begin{center}
\small
\begin{tabular}{||l l l l l l l l l l||}

\hline
2311&	2568&	2702&	2802&	2351&	2248&	2496&	2552&	2061&	1780	
\\
\hline
1453&	1823&
2415&	2596&	2885&	2947	&2936	&3113&	2877&	2974	
\\
\hline
3511&	3954&	2339	&1525&
1527	&2071	&2907&	2545&	1965&	2310	

\\
\hline
3064&	2455&	1515&	1516&	1527&	848&
1837&	2110	&1064	&975	

\\
\hline
1084	&2163	&4603&	6409	&6157&	5246&	4998	&5610&
3880&	2292	
\\
\hline
4130&	5911&	8518&	8876&	9693&	11377&	10841&	12442&	14649	&16939

\\
\hline
15266	&15964	&15842	&16711	&17760	&17274	&15813	&14470	&13623	&13521	

\\
\hline
14326&	14867&
16331	&16070	&15880&	18072&	21312	&22709	&24026	&22112	

\\
\hline
21492&	21567	&16331	&22416&
23215	&23309&	21192&	17992&	17550&	15295	

\\
\hline
13795	&11567	&9122&	7488	&6425	&4745&&&&
\\
\hline

\end{tabular}
\caption{ Values of the successive 30 minute  precipitation data on the selected region in the Brisbane area for 25 and 26 January 2013 as
millimeters.}
\label{table:3}
\end{center}
\end{table}

\begin{figure}[!tbp]
  \centering
  \begin{minipage}[b]{0.485\textwidth}
    \includegraphics[width=\textwidth]{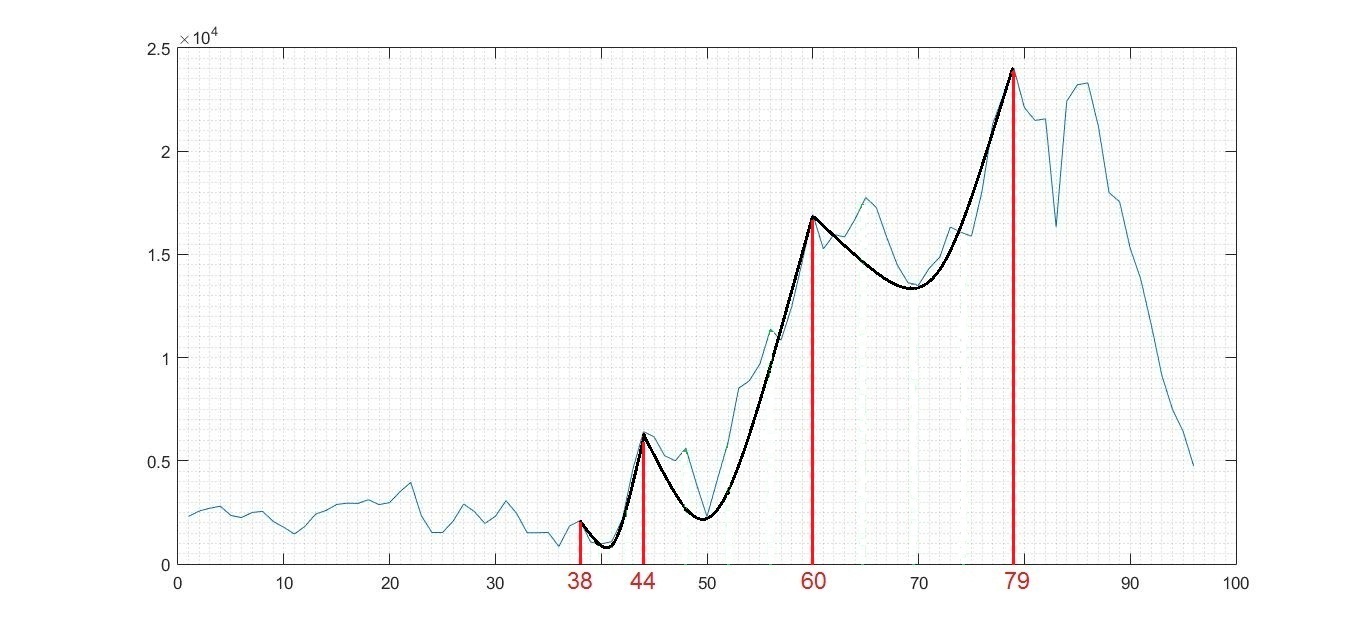}
    \caption{\footnotesize{The fitted black curve shows the scale intervals of 30 minutes precipitation on 25th and 26th January 2013. Scale intervals are shown by red lines.}}
  \end{minipage}
  \hfill
  \begin{minipage}[b]{0.49\textwidth}
    \includegraphics[width=\textwidth]{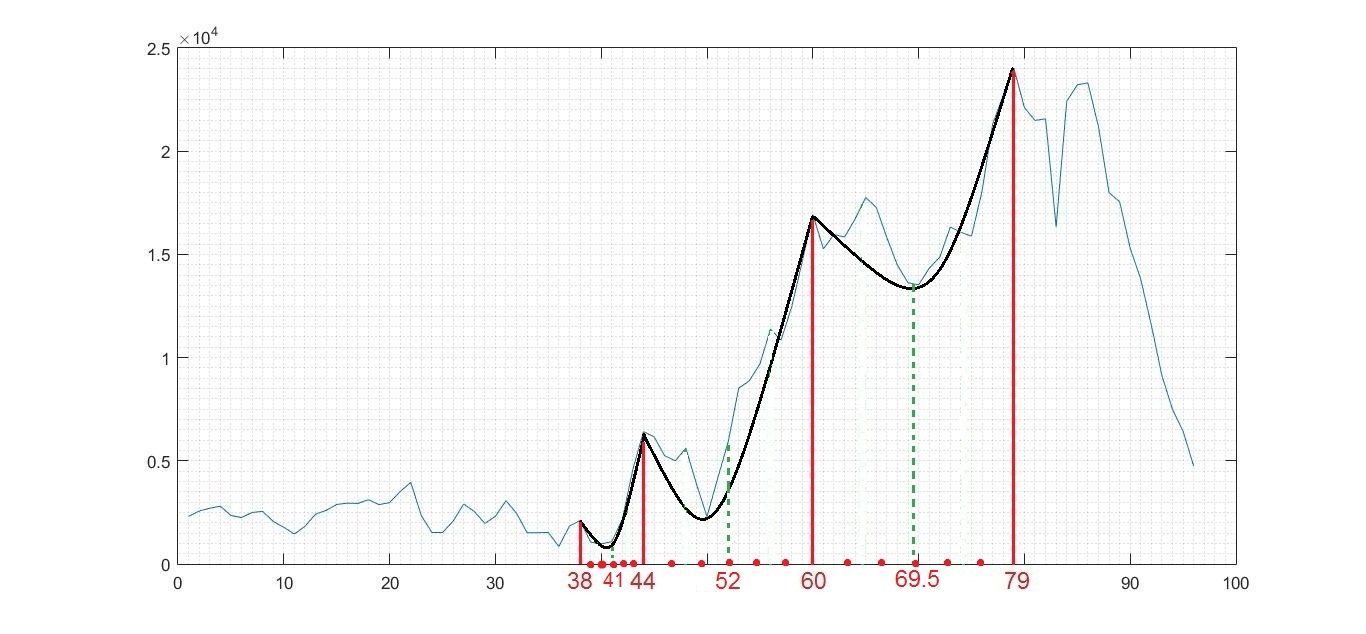}
    \caption{\footnotesize{Each scale interval, for  successive 30 minutes precipitation on 25-26 January 2013, is divided into two  equally length subintervals by green dashed lines.}}
  \end{minipage}
\end{figure}

\noindent
In Figure 7, precipitation per successive 30 minutes duration  on the two days are plotted by blue.
We fit some parabolas to the successive scale intervals in this figure and highlight the boundary of these scale intervals with  red lines. So we have three scale intervals where the end points of these scale intervals are $c_1=38$, $c_2=44$, $c_3=60$, $ c_4=79$.
Thus, the scale parameter can be evaluated by  $\lambda_{3,n-1}=\frac{c_{n+1}-c_n}{c_n-c_{n-1}}$ for $n=2,3$.
So, we have $\Lambda_3:=(\lambda_{3,1},\lambda_{3,2})=(2.667, 1.187)$.
We divide each scale interval into two equal subintervals which are indicated with green dashed lines in Figure 8 and partition these subinterval with the equally spaced red points.
{
Following the same method described at the beginning of the subsection 4.2, the  Hurst parameters related to the $m$-th subintervals of the $n$-th and $(n+1)$-th  scale intervals
  is estimated  by}
\begin{equation*}
{\mathcal{H}}_{(n,m)3}=\frac{\log(SS_{(n+1,m),3}/SS_{(n,m),3})}{2\log\lambda_{3,m}},
\end{equation*}
{ where
$SS_{(n,m),3}=\frac{1}{3}\sum_{k=1}^{3}z^2_{(n,m)k}$,
in which  $z_{(n,m)1}, z_{(n,m)2}, z_{(n,m)3}$  are  accumulated precipitation on successive partitions of $m$-th subinterval of the $n$-th  scale interval (presented in Table 8), which their ratio and corresponding Hurst estimates are shown in Table 9.

}
{\color{black}
\begin{table}[h!]
\begin{center}
\small
\begin{tabular}{||c c c c c||}
\hline
n &\vline &  $z_{(n,1)k}$    &\vline & $z_{(n,2)k}$  \\ [0.5ex]
 \hline\hline
1 &\vline &
1064	\:-\:975	\:-\:1084
 &\vline&
2163	\:-\:4603	\:-\:6409
\\
 \hline
2 &\vline&
14702\:-\:11946\:-\:11577
 &\vline&
23791	\:-\:29619	\:-\:39924
 \\
 \hline
3 &\vline&
49746	\:-\:54290	\:-\:45448
 &\vline&
46732	\:-\:52502	\:-\:71119
\\
 \hline
\end{tabular}
\caption{ precipitation values on partitions in subintervals.}
\label{table:8}
\end{center}
\end{table}

\begin{table}[h!]
\begin{center}
\small
 \begin{tabular}{||c c c c c c  ||}
 \hline
 n  &\vline & $\frac{SS_{(n+1,1),3}}{SS_{(n,1),3}}$& $\hat{\mathcal{H}}_{(n,1)3}$&  $\frac{SS_{(n+1,2),3}}{SS_{(n,2),3}}$& $\hat{\mathcal{H}}_{(n,2)3}$\\ [0.5ex]
 \hline\hline
 1 &\vline & 151.28&2.56& 45.37 &1.94\\
 \hline
 2 &\vline& 15.19& 7.93  &3.29& 3.48 \\
 \hline

\end{tabular}
\caption{ The ratios of squared of quadratic means and  corresponding  Hurst values.}
\label{table:9}
\end{center}
\end{table}

}

Thus
  $$\hat{{\bf H}}_{3}:=(\hat{\mathcal{H}}_{(1,1)3},\hat{\mathcal{H}}_{(1,2)3}, \hat{\mathcal{H}}_{(2,1)3},\hat{\mathcal{H}}_{(2,2)3})=(2.56, 1.94, 7.93, 3.48)$$
  ,
\begin{equation*}
\hat{{H}}_{3,1}=\frac{\hat{\mathcal{H}}_{(1,1)3}+\hat{\mathcal{H}}_{(1,2)3}}{2}=2.25, \:\:\:\:\:\:\hat{{H}}_{3,2}=\frac{\hat{\mathcal{H}}_{(2,1)3}+\hat{\mathcal{H}}_{(2,2)3}}{2}=5.7
\vspace{5mm}
\end{equation*}
and
\begin{equation*}
\hat{{{H}}}_{3}=\frac{\hat{{{H}}}_{3,1}+\hat{{{H}}}_{3,2}}{2}=\frac{2.25+5.7}{2}=3.98.
\vspace{5mm}
\end{equation*}
\noindent
{\bf Simple fractional Brownian sheet : Parameter Estimation}\\
{
Here we are to justify the structure of simple fractional Brownian sheet (sfBs), defined by (\ref{TS}), as a MSI field for the precipitation data
 of Figure 2.
After detecting scale rectangles, the scale parameters $(\lambda_1, \lambda_2)$ are estimated by (\ref{s1}).
 By (\ref{s2}), we have that $\hat{{H}}_{1}=1.435$ and $\hat{{H}}_{2}=1.765$.
As we have
$m_1=14$ equally spaced samples in each horizontal scale intervals in Figure 5, and $m_2=10$ equally spaced samples in each vertical scale intervals in Figure 6. These samples provide some partitions for each interval. Let $x_{1(n,k)}$ be the sum of precipitation on the $k$-th vertical strips corresponding to the $k$-th partition  of the $n$-th horizontal scale interval where  $n=1,2,3$ and  $k=1,\ldots,14$. Now we are to estimate the Hurst parameters of the two-dimensional fractional Brownian sheet  $ H'_{1}$ and $ H'_{2}$ corresponding to sfBs (2.5). For this we consider quadratic variations of lengths two and one  for horizontal scale intervals as
$$
SS_{1(n,2)}=\frac{1}{[\frac{m_1}{2}]-1}\sum_{l=2}^{[\frac{m_1}{2}]}(x_{1(n,2l)}-x_{1(n,2l-2)})^2,
\:\:\:\:\:\:
SS_{1(n,1)}=\frac{1}{[\frac{m_1}{2}]-1}\sum_{l=2}^{[\frac{m_1}{2}]}(x_{1(n,l)}-x_{1(n,l-1)})^2
$$
and  for vertical scale intervals as
  $$
SS_{2(n,2)}=\frac{1}{[\frac{m_2}{2}]-1}\sum_{l=2}^{[\frac{m_2}{2}]}(x_{2(n,2l)}-x_{2(n,2l-2)})^2,
\:\:\:\:\:\:
SS_{2(n,1)}=\frac{1}{[\frac{m_2}{2}]-1}\sum_{l=2}^{[\frac{m_2}{2}]}(x_{2(n,l)}-x_{2(n,l-1)})^2,
$$
where $x_{2(n,k)}$ is the accumulated precipitation on the $k$-th horizontal strips corresponding to the $k$-th partition  of the $n$-th vertical scale interval; $n=1,2,3$ and $k=1,\ldots,10$.
\\
By component-wise self-similarity with Hurst $H_i^{'}$ and component-wise stationary increment property of the samples inside each scale rectangle of sfBs; this is by the fact that when we consider accumulate precipitation on the strips and denote them by $x_{i(n,l)}$ they construct a simple fractional Brownian motion which inside each scale rectangles samples are just some multiple of a fractional Brownian motion and so have stationary increments,  so we have that
\begin{equation}\label{oo}
(x_{i(n,2l)}-x_{i(n,2l-2)})\overset{\mathcal{L}}{=}2^{H_i^{'}}(x_{i(n,l)}-x_{i(n,l-1)})
\end{equation}
 for $i=1,2$, and by MSI behavior of sfBs (2.5) we have that $x_{i(n,2l)}\overset{\mathcal{L}}{=}\lambda_{i}^{H_i}x_{i(n-1,2l)}$. So
 $$SS_{i(1,2)}\overset{\mathcal{L}}{=}\frac{SS_{i(2,2)}}{{\lambda_i}^{2H_i}}\overset{\mathcal{L}}{=}\frac{SS_{i(3,2)}}{{\lambda_i}^{4H_i}}, \;\;\;\;\;\; SS_{i(1,1)}\overset{\mathcal{L}}{=}\frac{SS_{i(2,1)}}{{\lambda_i}^{2H_i}}\overset{\mathcal{L}}{=}\frac{SS_{i(3,1)}}{{\lambda_i}^{4H_i}}$$ and by
 (\ref{oo}), $SS_{i(1,2)}\overset{\mathcal{L}}{=}2^{2H_i^{'}}SS_{i(1,1)}$. Therefore by assuming
$$
U_{i2}=SS_{i(1,2)}+\frac{SS_{i(2,2)}}{{\lambda_i}^{2H_i}}+\frac{SS_{i(3,2)}}{{\lambda_i}^{4H_i}},
\:\:\:\:\:\:\:\:\:\:
V_{i1}=SS_{i(1,1)}+\frac{SS_{i(2,1)}}{{\lambda_i}^{2H_i}}+\frac{SS_{i(3,1)}}{{\lambda_i}^{4H_i}},\:\:\: i=1,2,
$$
we have that $U_{i2}\overset{\mathcal{L}}{=}2^{2 H_{1i}^{'}}V_{i1}$. Also we note that for $i=1,2$, $U_{i2}$ and $V_{i1}$ consist of $r_i([\frac{m_i}{2}]-1)$
increment samples where $r_1=r_2=3$ is the number of vertical or horizontal scale intervals and $m_1=14$ is the number of vertical strips in each horizontal scale interval and $m_2=10$ is the number of horizontal strips in each vertical scale interval.
So by similar method to the theorem 1 of Rezakhah et al. \cite{R-M3} where used the result of Ayache et al \cite{Ayache} for fBm, we have the estimation of the Hurst parameters of sfBs as
$$
\hat H_{1}^{'}=\frac{\log(\frac{\hat{U}_{12}}{\hat{V}_{11}})}{2\log2}=0.47,
\:\:\:\:\:\:\:\:\:\:\:\:
\hat H_{2}^{'}=\frac{\log(\frac{\hat U_{22}}{\hat V_{21}})}{2\log2}=0.91, 
$$
where
$$
\hat U_{i2}=SS_{i(1,2)}+\frac{SS_{i(2,2)}}{{\lambda_i}^{2\hat H_i}}+\frac{SS_{i(3,2)}}{{\lambda_i}^{4\hat H_i}},
\:\:\:\:\:\:\:\:\:\:
\hat V_{i1}=SS_{i(1,1)}+\frac{SS_{i(2,1)}}{{\lambda_i}^{2\hat H_i}}+\frac{SS_{i(3,1)}}{{\lambda_i}^{4\hat H_i}}.
$$
\noindent
{

\subsection{\bf Prediction }
For prediction we use component-wise (latitude and longitude) scale invariant property of sfBs model (2.5) assumed for the precipitation in Brisbane  area of Australia on the specified duration of time (25 and 26 January 2013)  depicted in Figure 2 to predict the values of
precipitation on different region ($A_{ij}$) based on the values of the precipitation on the first region ($A_{11}$) in Figure 9.
Here we study prediction of the precipitation in surface which are depicted in  Figures 9 and 10.
For this, first we denote the 9 scale rectangles as  $A_{11},\ldots, A_{33}$  in Figure 9,  which are in correspondence  to the vertical and horizontal scale intervals presented in Figures 3 and 4. We partition each scale rectangle $A_{ij}$ into four sub-rectangles $A_{ij,k}$; $k=1,2,3,4$  in Figure 10 that are related to the sub-intervals in Figures 5 and 6. {The  scale sub-rectangles are shown in Figure 10  with  black lines along vertical axis at points $d_1=0$, $d_2=7$, $d_3=14$, $d_4=22.5$, $d_5=31$, $d_6=41.5$, $d_7=52$ and along horizontal axis at points  $e_1=0$, $e_2=5$, $e_3=10$, $e_4=16.5$, $e_5=23$, $e_6=31.5$, $e_7=40 $.  
Let $(s_1, t_1), (s_2, t_1),(s_1,t_2) $ and $(s_2, t_2)$ denote the center points of the sub-rectangles $A_{11,1},  A_{11,2}, A_{11,3}, A_{11,4}$.  So for $i,j=1,2,3$  we denote the accumulated precipitation  on sub-rectangles $A_{ij,1}, A_{ij,2},A_{ij,3}$ and $A_{ij,4},$ which are recorded in Table 10,  by  $\,X(\lambda_1^{i-1}s_1, \lambda_2^{j-1}t_1),\; X(\lambda_1^{i-1}s_2, \lambda_2^{j-1}t_1),\; X(\lambda_1^{i-1}s_1, \lambda_2^{j-1}t_2)\;$ and $X(\lambda_1^{i-1}s_2, \lambda_2^{j-1}t_2),$   related to their center points, respectively.
 Reffering to the Figure 10, we assume that the correlation of the accumulated precipitations on the  corresponding sub-rectangles in the first and second vertical scale strips, say $(A_{1j,k}, A_{2j,k})$, $\,j=1,2,3$ and $k=1,2,3,4$  to be the same. Also the correlation of the accumulated  precipitations on corresponding  sub-rectangles of the second and third vertical scale strips $(A_{2j,k}, A_{3j,k})$ to be the same. Finally the accumulated precipitation on corresponding sub-rectangles of the first and third vertical scale strips  $(A_{1j,k}, A_{3j,k})$ have the same correlation.  So using the notations of Remark 6 we estimate $R_{X_1^*}(2)$  by calculating the sample correlation of  the accumulated precipitations on  the pair of sub-rectangles $(A_{1j,k}, A_{3j,k})$ as  $0.9046$.    The    $R_{X^*_1}(1)$  can be estimated as the correlation of  the accumulated precipitation on the  pair of sub-rectangles $A_{1j,k}, A_{2j,k}$  or on the pair of sub-rectangles $A_{2j,k}, A_{3j,k}$ {which are corresponding sub-rectangles in the first and second horizontal scale strips and in the second and third horizontal scale strips respectively and} are evaluated as $-0.9573$ and $-0.9544$.
}

\begin{figure}[!tbp]
  \centering
  \begin{minipage}[b]{0.49\textwidth}
    \includegraphics[width=\textwidth]{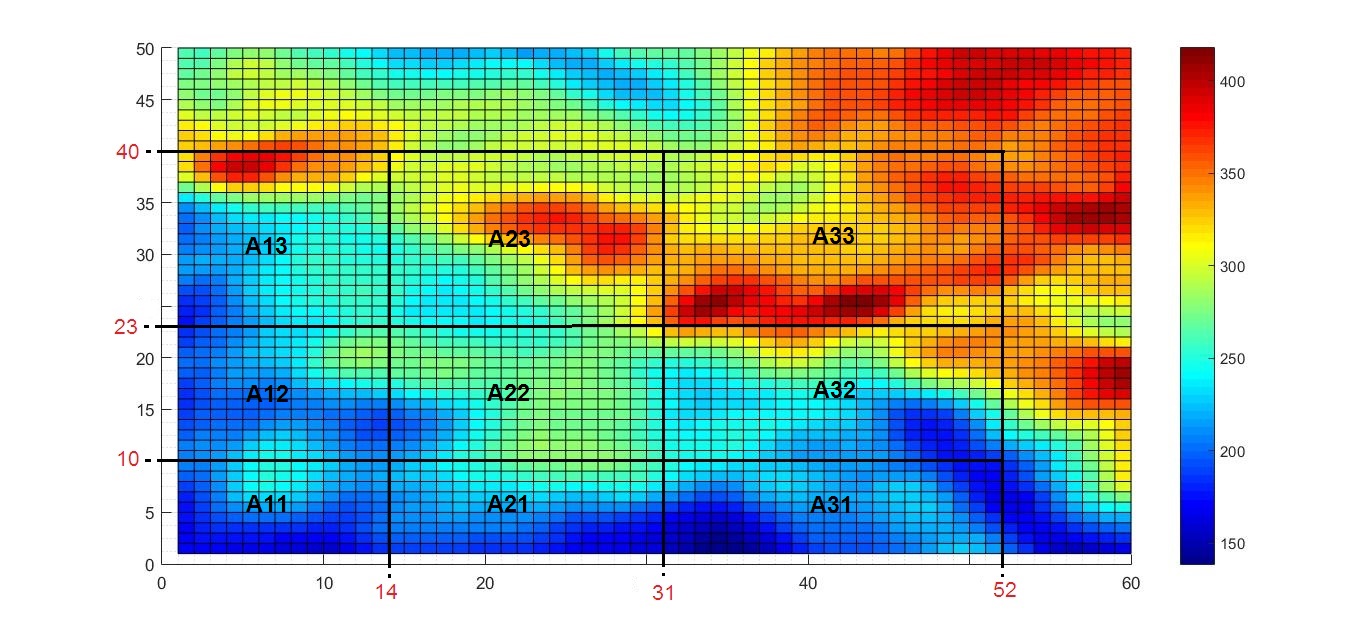}
    \caption{\tiny{Nine rectangle areas, which are obtained by crossing lines along end points of scale intervals.}}
  \end{minipage}
  \hfill
  \begin{minipage}[b]{0.5\textwidth}
    \includegraphics[width=\textwidth]{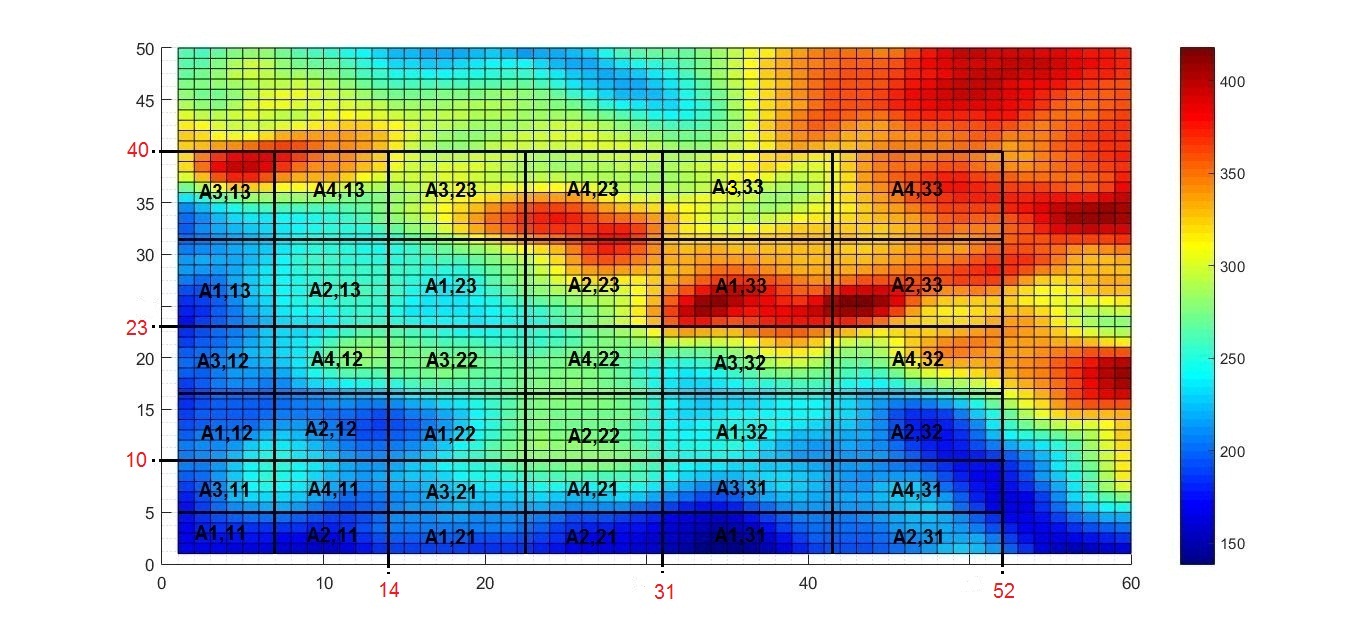}
    \caption{ \tiny{sub-rectangle areas, which are obtained by crossing lines along subinterval points of scale intervals.}}
  \end{minipage}
\end{figure}

\begin{table}[h!]
\begin{center}
\small
 \begin{tabular}{||c c c c c c c c c c c||}
 \hline
sub-rectangular area &\vline &$ A_{11,1}$ & $A_{11,2}$ & $A_{11,3}$ &  $A_{11,4}$ & $A_{12,1}$& $A_{12,2}$& $A_{12,3}$& $A_{12,4}$ & $A_{13,1}$\\ [0.5ex]
 \hline
precipitation value  &\vline&  6451 & 6590 & 7816 & 7701 & 9314 & 9577&9572&11686&12983\\
 \hline
  \hline
sub-rectangular area &\vline & $A_{13,2}$ & $A_{13,3}$ &
$ A_{13,4}$ &  $A_{21,1}$ & $A_{21,2}$& $A_{21,3}$& $A_{21,4}$& $A_{22,1}$ & $A_{22,2}$\\
 \hline
precipitation value &\vline&  14915 & 18062 & 17882 & 8631 & 7761 & 10216 &10733&13567&14716\\
 \hline
  \hline
sub-rectangular area &\vline & $A_{22,3}$ & $A_{22,4}$ & $A_{23,1}$ & $ A_{23,2}$ & $A_{23,3}$& $A_{23,4}$& $A_{31,1}$& $A_{31,2}$ & $A_{31,3}$\\
 \hline
precipitation value &\vline&  14370 & 15161 & 19372 & 23591 & 22272 & 22617 &9300&10910&11715\\
 \hline
  \hline
 sub-rectangular area &\vline & $A_{31,4}$ & $A_{32,1}$ & $A_{32,2}$ &  $A_{32,3}$ & $A_{32,4}$& $A_{33,1}$& $A_{33,2}$& $A_{33,3}$ & $A_{33,4}$\\
 \hline
precipitation value &\vline&  10428 & 16210 & 15111 & 20707 & 21592 & 31390 &31123&27169&31183\\
 \hline
\end{tabular}
\caption{ The accumulated precipitation value on the sub-rectangular regions.}
\label{table:4}
\end{center}
\end{table}

{
Now following Remarks 5 and 6 we show that the precipitations has component-wise scale Markov property. 
Following notations of  Remark 6, and the estimated values of $R_{X_1^*}(2)$ and $R_{X_1^*}(1)$  we estimate    $\alpha_{Y_1}(2)$  as the  partial
auto-correlation function of $Y_1(\cdot )$  at lag 2  as $-0.1416$ and $-0.0707$ based on the two estimated values of
 $R_{X_1^*}(1)$ . As number of samples is 12 and  both of these values are between $\pm 1.96/\sqrt{12}=\pm 0.5658$. Therefore, at level $\%95$  the first component  scale Markov property  for the precipitation is accepted.  By the same method we evaluate the sample partial auto-correlation $\alpha_{Y_2}(2)$    of $Y_2$ at lag 2 is evaluated by estimating  $R_{X_2^*}(2)$ and  $R_{X_2^*}(1)$
which are estimated by  using  the accumulated precipitation on the pair of sub-rectangles  
$(A_{i1,k},A_{i3,k})$  and on pair of sub-rectangles  $(A_{i1,k},A_{i2,k})$ or $(A_{i2,k},A_{i3,k})$   for $i=1,2,3$ and $k=1,2,3,4$ which  cause the sample partial auto-correlation at lag 2, say $\alpha_{Y_2}(2)$,  to be evaluated as $0.3242$ or  $0.2624$  which both are between $\pm 1.96/\sqrt{12}=\pm 0.5658$. So  at level $\%95$ it is accepted that $Y_2$ follows an AR(1) model.
Thus by Remark 5 and Remark 6 of Section 3 the sfBs $X(t_1,t_2)$ of the accumulated precipitation on these sub-rectangles  has component-wise scale Markov property.}
{ For the simplicity let $Y_{ij,k}$  denote the accumulated precipitation on sub-rectangle $A_{ij,k}$  for $ i,j = 1,2,3 $ and $k = 1,2,3,4$.
   So, the accumulated precipitation $Y_{ij,k}$  have component-wise Markov property in components $i$ and $j$ for fixed $k$.
Also following   MSI property, ${Y}_{ij,k}\overset{\mathcal{L}}{=}\lambda_1^{(i-1){H}_1}\lambda_2^{(j-1){H}_2} Y_{11,k}$  for  $i,j=1,2,3$ and $k=1,2,3,4$ where  $Y_{ij,k}$ as the accumulated precipitation on sub-rectangles  $A_{ij,k}$ are shown in Figure 10 and} { their values are recorded in Table 10.}
Therefore  under the assumption that the precipitation  $Y_{11,k}$ is known, $Y_{ij,k}$
 can be predicted using conditional expectation as
\begin{equation}\label{r12}
\hat{Y}_{ij,k|{11,k}}=\hat {E}[{Y}_{ij,k}|Y_{11,k}]=\hat\lambda_1^{(i-1){\hat H}_1}\hat\lambda_2^{(j-1){\hat H}_2} Y_{11,k}
\end{equation}
for $i,j=1, 2, 3$ and $k=1,2,3,4$. 
Let  $ Y_{ij}=\sum_{i=k}^4Y_ {ij,k}$  be the accumulated precipitation  on the scale rectangular  $A_{ij}$ for  $i,j = 1,2,3$  in Figure 9.
 So the prediction of $Y_{ij}$  provided the precipitation $Y_{11,k}; k=1,2,3,4$ is known can be evaluated as $\sum_{k=1}^4 \hat{Y}_{ij,k|{11,k}}$ (Table 11). The real precipitation on $A_{ij}$ ($Y_{ij}$) and corresponding predicted value are plotted in Figure 11.

\begin{figure}[!tbp]
  \centering
  \begin{minipage}[b]{0.49\textwidth}
    \includegraphics[width=\textwidth]{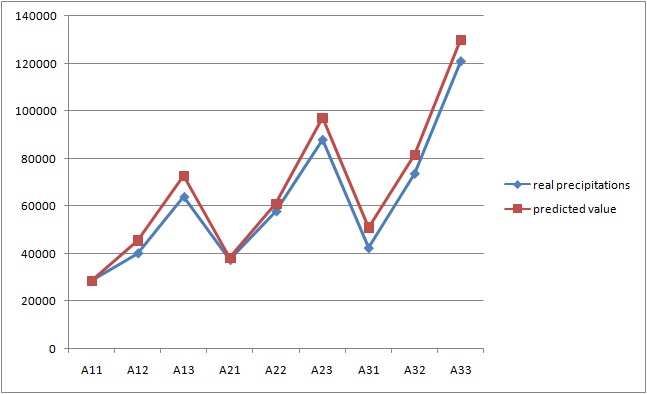}
     \end{minipage}
    \caption{The accumulated precipitation value  and corresponding predicted value  on  rectangular regions.}
\end{figure}
\begin{table}[h!]
\begin{center}
\small
 \begin{tabular}{||c c c c c c c c c c c||}
 \hline
rectangular area \!\!\!\!&\vline & $A_{11}$ & $A_{12}$ & $A_{13}$ &  $A_{21}$ & $A_{22}$& $A_{23}$& $A_{31}$& $A_{32}$ & $A_{33}$\\ [0.5ex]
 \hline
precipitation values \!\!\!\! &\vline&  28558 & 40149 & 63842 & 37341 & 57814 & 87852&42353&73620&120865\\
 \hline
predicted values  &\vline&  28558 & 45562 & 72691 & 38168 & 60894 & 97151&51011&81384&129842\\
  \hline
\end{tabular}
\caption{ The accumulated precipitation value and corresponding predicted value on  rectangular regions.  }
\label{table:5}
\end{center}
\end{table}
The mean absolute percentage error (MAPE)  is a famous measure for the prediction accuracy, \cite{Bowerman, Hanke, Moren}, is defined by
$$
\gamma=\frac{1}{n}\sum_{i=1}^n
|\frac{r_i-\hat{r}_i}{r_i}|\times 100,
$$
where $r_i$ and $\hat{r}_i$ are respectively real value and predicted value for $i$-th data point and $n$ is  the number of data points.
According to interpretation of MAPE values by Lewis\cite{Lewis},
for highly accurate forecasting  $\gamma\leq10$,
 good forecasting  $10<\gamma\leq20$,
 reasonable forecasting  $20<\gamma\leq50$
and for inaccurate forecasting  $\gamma>50$, see Moreno et  al.  \cite{Moren}.
\\
 So we consider MAPE as an accuracy index for the  predicted values
 on eight rectangular areas in Figure 9.
The  MAPE value  for these data is evaluated as
\begin{equation}\label{g12}
\gamma^*=\frac{1}{8}\sum_{k=1}^3\sum_{l=1}^3
|\frac{{Y}_{kl}-\hat{Y}_{kl}}{Y_{kl}}|\times 100,
\end{equation}
where $\hat{Y}_{11}$ is equal to $Y_{11}$ and is ignored in calculations, because the rectangular region $A11$ is the initial region.  Table 12 shows the absolute values in (\ref{g12}).
 \begin{table}[h!]
\begin{center}
\small
 \begin{tabular}{||c c c c c c c c c c c||}
 \hline
rectangular area &\vline & $A_{11}$ & $A_{12}$ &$ A_{13}$ &  $A_{21}$ & $A_{22}$& $A_{23}$& $A_{31}$&$ A_{32}$ & $A_{33}$\\ [0.5ex]
 \hline
absolute value  &\vline&  0 & 0.135 & 0.139 & 0.022 & 0.053 & 0.106&0.204&0.105&0.074\\
 \hline
\end{tabular}
\caption{ Absolute of difference between  the actual   precipitation values   and the corresponding predicted values
 are divided by the actual precipitation values.  }
\label{table:12}
\end{center}
\end{table}
These  values  are absolute of difference between  the actual  accumulated precipitation values on nine rectangular regions  and the corresponding predicted values
 divided by the actual accumulated precipitation values.
 The MAPE value is obtained as $\gamma^*=10.5$.
\\
Hence, by Lewis's classification for MAPE values, this is a verified certificate for the accuracy of our prediction method   of precipitation values.

One could follow the method of this section to predict the precipitation in time while we have the same circumstances (as the DSI behavior always valid for restricted duration) by having the precipitation in some scale interval of time. Prediction can be followed in surface and time simultaneously by applying these predictions successively.
}
{
\section{Discussion $\&$ Conclusions}
In this paper, we have introduced multi-scale invariant (MSI) fields which have  component-wise discrete scale invariant property.
It is shown that the covariance function of the MSI field with Markov property (MMSI)  is characterized by the covariance functions and variances of samples  on the first scale rectangle. A two-dimensional simple fractional Brownian sheet (sfBs) as an example of MSI field is demonstrated and applied to a set of real data, say precipitation in Brisbane area of Australia and considered prediction that its high accuracy is shown by using MAPE method.\\
Regarding Markov property of MSI field, even though assuming Markov property for the precipitation on different parts of some area is not so realistic but this can be a property in some other MSI fields.  Results of section 3 enables one to evaluate the covariance structure between samples of any two scale rectangles of Markov MSI field (3.4) by using the covariance structure and variance functions of samples inside first scale interval.

}
\noindent
\\\\
{\Large{\bf Acknowledgment }}\\
{
We would like to express our sincere thanks to the two anonymous reviewers that their valuable comments helped us to improve the quality of this manuscript.
}
 The authors would like to express their thanks to  Professor Alan Seed from Australian Bureau of Meteorology for providing the data used in this paper.

\end{document}